\documentclass[12pt,a4paper]{article}
\usepackage[utf8]{inputenc}

\usepackage{booktabs}
\usepackage{natbib}
\setcitestyle{numbers,open={[},close={]}}

\ifdefined\AEJ
\usepackage[margin=1.5in]{geometry}
\else
\usepackage[hmargin=2.5cm,vmargin=3cm]{geometry}
\fi

\usepackage{comment}
\usepackage{xcolor}
\usepackage{tikz}
\usepackage{color}
\usepackage{amsmath}
\usepackage{amsthm}
\usepackage{amssymb}
\usepackage{amsfonts}
\usepackage{geometry}
\usepackage[normalem]{ulem}
\usepackage{float}

\usepackage{hyperref}
\hypersetup{colorlinks=true,linkcolor=blue,citecolor=blue}

\usepackage{mathtools}
\usepackage{tikz}
\usepackage{pgfplots}
\usepgfplotslibrary{fillbetween}

\theoremstyle{plain}
\newtheorem{theorem}{Theorem}[section]

\newtheorem{corollary}[theorem]{Corollary}
\newtheorem{lemma}[theorem]{Lemma}

\theoremstyle{definition}
\newtheorem{definition}[theorem]{Definition}

\newtheorem{example}[theorem]{Example}

\newtheorem{assumption}[theorem]{Assumption}

\DeclareMathOperator{\cav}{cav}

\title{Bayesian Persuasion with Mediators\footnote{We are grateful to Alexander Bloedel, Alessandro Pavan, and seminar participants at Bocconi University, Stony Brook Workshop on Strategic Communication and Learning, and Tel Aviv University for their comments.}}

\author{ \large Itai Arieli\thanks{Technion. Email:  iarieli@technion.ac.il. Itai Arieli was partially supported by the Israel Science Foundation grant \#2030524.} \ \ \ \ Yakov Babichenko\thanks{Technion. Email:  yakovbab@technion.ac.il. Yakov Babichenko was partially supported by the United States - Israel Binational Science Foundation BSF grants \#2018397 and \#2021680.} \ \ \ \ Fedor Sandomirskiy\thanks{Caltech. Email:  fsandomi@caltech.edu. Fedor Sandomirskiy thanks Linde Institute at Caltech and the National Science Foundation (grant  CNS 1518941)
for their support.}}
\date{}

\newcommand{\argmax}{{\mathrm{argmax}}}
\newcommand{\supp}{{\mathrm{supp}}}

\newcommand{\dd}{{\mathrm{d}}}

\newcommand{\p}{{p}}

\ifdefined\DRAFT

\definecolor{ForestGreen}{rgb}{.13,.54,.13}
\definecolor{violet}{cmyk}{0.79,0.88,0,0}
\newcommand{\fed}[1]{{\color{ForestGreen}{(\textbf{Fedor:} #1)}}}

\newcommand{\ed}[1]{{{\color{red} {#1}}}}

\else
\newcommand{\fed}[1]{}
\newcommand{\yakov}[1]{}
\newcommand{\ed}[1]{#1}

\fi

\begin{document}

\maketitle

\begin{abstract}
An informed sender communicates with an uninformed receiver through a sequence of uninformed mediators; agents' utilities depend on receiver's action and the state. For any number of mediators, the sender's optimal value is characterized. For one mediator, the characterization has a geometric meaning of constrained concavification of sender's utility, optimal persuasion requires the same number of signals as without mediators, and the presence of the mediator is never profitable for the sender. Surprisingly, the second mediator may improve the value but optimal persuasion may require more signals.
\end{abstract}

\section{Introduction}

Strategic information transmission has been studied
for decades in sender--receiver models including cheap talk (e.g.,   Crawford
and Sobel \cite{crawford1982strategic}), Bayesian persuasion (Aumann et al.~\cite{aumann1995repeated}; Kamenica and
Gentzkow \cite{kamenica2011bayesian}), models with multiple senders (e.g., Gentzkow and Kamenica \cite{gentzkow2016competition}; Bhattacharya and Mukherjee \cite{bhattacharya2013strategic}), and settings with multiple receivers
(e.g., Arieli and Babichenko \cite{arieli2019naive}; Alonso and Camara \cite{alonso2016persuading}; Arieli et al. \cite{arieli2021feasible}). The classic
sender--receiver model focuses on a two-agent interaction: the sender holds some private information and decides what information to communicate to the receiver. The receiver then takes an action based on the information {she} received from the sender and this action affects the utility of both the sender and the receiver.
At this point, the literature has developed a quite complete understanding of strategic information transmission where the sender communicates with the receiver directly.

In many realistic scenarios, an informed party communicates with an uninformed {decision-maker indirectly through mediators.} {For example, the board of directors shares information with the CEO, the CEO transmits it to a team leader, who then disseminates it among a group of engineers; however, the  CEO and the team leader, playing the role of mediators, may pursue their own goals---e.g., shaped by their performance metrics---and so the original transmitted information gets strategically distorted.}   Politicians share information with {mass media, which may have its own political agenda and decide to publish only that part of the information that conforms to its agenda, and then---in authoritarian regimes---this content may be further filtered or even blocked by censorship before reaching the public.} 
{In supply chains,} producers share information with vendors who partially disclose it to buyers, yet the vendors' incentives may be misaligned with the producers' incentives, e.g.,  if they split warranty responsibilities and revenue unequally.

The goal of this work is to study strategic communication in the presence of mediators and hence to go beyond the standard sender--receiver setting, where the information flow is controlled by one agent. {We extend the standard Bayesian persuasion model by adding a sequence of $n$ mediators $M_1,\ldots,M_n$ between the sender $S$ and the receiver $R$:
$$S\to M_1\to M_2\to\ldots \to M_n\to R.$$
The set of states $\Omega$ is finite and {the only agent observing the realized state $\omega\in\Omega$ is the sender;
mediators can screen only the information received from their predecessor and cannot generate new information. 
Each mediator has her own objective, which may differ from the objectives of the sender and the other mediators; it determines which information the mediator decides to pass on to her successor and which information to conceal.
Before the state is realized, the sender commits to an information revelation policy by which she will reveal information to the first mediator $M_1$; then $M_1$, observing the policy chosen by the sender, selects the policy that she will use to transmit the information to $M_2$, and so on; the last mediator $M_n$ selects a policy that will be used to reveal information to the receiver. {Once the sender and all the mediators have chosen their policies,} the state is realized and the information is revealed {sequentially} according to {these policies.} {Finally,} the receiver takes an action, and {each agent receives her utility,} which depends on the action and the realized state.}}
What is the sender's optimal value, i.e., the best payoff to the sender in a subgame perfect equilibrium? How is this {value} affected by the presence of the mediators? 
We propose a general tractable geometric approach to answer these questions.

As is common in literature on Bayesian persuasion, we let the indirect utility function of the sender  be the function $v_S:\Delta(\Omega)\rightarrow\mathbb R$, where $v_S(p_R)$ is the expected payoff to the sender when  the  {receiver holding a posterior belief $p_R\in \Delta(\Omega)$ plays her best reply action. The indirect utility $v_{M_i}$ of a mediator $M_i$ is defined similarly.} 

\ed{Our characterization of the sender's optimal value relies on the following notion. 
We say that beliefs $q_1,\ldots,q_{{|\Omega|}} \in \Delta(\Omega)$ are \emph{affine dominating} with respect to a function $f:\Delta(\Omega)\to \mathbb{R}$ if the hyperplane 
passing through the collection of points $\big(q_k,\,f(q_k)\big)_{k=1,...,{|\Omega|}}$ lies above the graph of $f$ in the region {given by the convex hull of $(q_1,...,q_{{|\Omega|}})$.}}
\smallskip

{For $n=1$ mediator, we demonstrate that the
sender's optimal value is given by the constrained concavification of her indirect utility $v_S$, where the {beliefs are constrained to be} affine dominating with respect to the utility function $v_{M_1}$ of the mediator (Theorem~\ref{theorem:main}); moreover,  $|\Omega|$ signals turn out to be enough for optimal persuasion.
This result shows that the concavification formula derived by Kamenica and Gentzkow~\cite{kamenica2011bayesian} for the no-mediator case can be generalized in the presence of a mediator by adding extra constraints, and that optimal persuasion with a mediator requires the same number of signals as direct persuasion.
Since the constraints can only decrease the optimal value, the presence of one mediator lowers the sender's value compared to direct persuasion.}

{For $n\geq 2$ mediators, there are new phenomena. Similarly to the one-mediator case, the sender can guarantee the constrained concavification of $v_S$, where {the posteriors are required to be} affine dominating with respect to the utility functions of all the mediators. However, the sender can improve upon this guarantee by taking into account the cancellation of mediators' incentives. The presence of successors can eliminate some of the predecessors' profitable deviations; e.g.,  information that a single mediator $M$ would find profitable to garble may be perfectly transmitted by $M$ as her garbling may trigger extra garbling by subsequent mediators, thereby eliminating $M$'s profit.}

{In Theorem~\ref{th_n_mediators}, we characterize the sender's optimal value for any number $n$ of mediators. The characterization remains geometric but becomes more involved compared to the case of one mediator; it requires an analog of the affine domination defined recursively and imposes constraints in the space $\Delta\big(\Delta(\Omega)\big)$ rather than $\Delta(\Omega)$. Despite these complications, the characterization remains tractable and useful for solving concrete problems. We find an explicit solution in an example with two mediators and see that the optimal persuasion may require more than $|\Omega|$ signals and that adding the second mediator may increase the sender's value compared to the one-mediator case.}

\subsection{Related work}\label{sec:related}
\ed{
\paragraph{Mediated persuasion.} Independently and concurrently, Zapechelnyuk~\cite{zapechelnyuk2022sequential} studied a model that is analogous to the single-mediator version of ours. 
In contrast to our paper, his goal was not to characterize the sender's optimal policy but to compare this setting to the one where the mediator --- instead of garbling --- can add more information; see a model by Li and Norman~\cite{li2018sequential} discussed below.
}

\paragraph{Persuasion of a rationally inattentive receiver.} {The one-mediator version of our model is related to the recent literature on Bayesian persuasion of a rationally inattentive receiver (Wei~\cite{wei2020persuasion}, Bloedel and Segal~\cite{bloedel2020persuading}). In this literature, the sender--receiver communication is direct, but the receiver incurs attention costs given by a convex function such as entropy or quadratic loss, which may incentivize her to garble the information obtained from the sender before processing it. The papers most related to ours are those of Lipnowski et al.~\cite{lipnowski2020attention,lipnowski2020optimal}. In their model, the inattentive receiver can be split into a mediator--receiver pair, where the mediator does the garbling, the receiver takes an action, and both get the same utility. \ed{Our analysis of the single-mediator case  relies on arguments similar to those of~\cite{lipnowski2020attention,lipnowski2020optimal}; see the discussion after Theorem~\ref{theorem:main}.} 
 } 

\paragraph{Mediated cheap talk and similar models.} {Kosenko~\cite{kosenko2018mediated} considers a problem of mediated information transmission with one mediator, who, unlike in our model, selects her policy simultaneously with the sender. In other words, the mediator decides on the interpretation of the sender's signal without being sure how this signal is going to be generated and that, in particular, whether the sender is going to trick her. This feature brings the problem closer to the cheap talk model of Crawford and Sobel~\cite{crawford1982strategic} (communication without  commitment) than to the Bayesian persuasion model of Kamenica and Gentzkow~\cite{kamenica2011bayesian} (communication with commitment). Like the cheap talk model, Kosenko's model always admits a non-informative equilibrium, which is not the case in our setting. Consistent with the common wisdom that cheap talk is less tractable than persuasion, Kosenko's model can be analyzed comprehensively only for a binary state under the restrictive assumption of binary signals. By contrast, our model turns out to be tractable for an arbitrary number of states, any number of mediators, and without any restriction on the signaling policies that can be used.

Several papers deal with concrete problems of mediated communication. Kuang et al.~\cite{kuanghierarchical} consider a model of mediated persuasion within an organization, where the receiver's decision is binary and mediators are driven by reputation concerns. Ivanov~\cite{ivanov2010communication} and Ambrus~\cite{ambrus2013hierarchical} introduce mediators to the classic cheap talk model of Crawford and Sobel with quadratic utilities.  Qian~\cite{qian2020jmp} analyzes a one-mediator model of censorship with partial commitment.
Levkun~\cite{levkun2020mediation} models elite--media--public information transmission where public takes a binary action.
In all these papers, unlike in ours, tractability comes at the cost of a particular functional form of utilities and restrictions on the feasible signaling policies and actions.} 

\paragraph{Multi-sender persuasion.} Within the literature on Bayesian persuasion with multiple senders, the model closest to ours is the model of sequential persuasion  of Li and Norman~\cite{li2018sequential} and Wu~\cite{wu2020essays}. Informed senders move sequentially and each can add extra information to the information already transmitted to the receiver by predecessors. 
In our paper, the dynamics and incentives are different since only the first agent (the sender) provides information, while all the successors (the mediators) can only garble it. Despite this difference, both models  share some similarity in recursive formulas for the first mover's value and both satisfy a version of the revelation principle. In Appendix~\ref{sect_partial_order}, we explain this similarity by embedding both models into a general class of sequential games over partially ordered sets. For simplicity, we do not pursue this general perspective in the rest of the paper.

%

\paragraph{Persuasion over communication networks.} Our results can be viewed as one of the first steps toward understanding Bayesian persuasion over directed communication networks, where the edges represent the direction of the information flow and intermediaries aim to affect the actions of decision-makers located at the leaves of the network.
Currently, such models are well understood when they have a single edge (single sender, single receiver; see, e.g., Kamenica and Gentzkow~\cite{kamenica2011bayesian}), a star network with outgoing edges from {the central node} (single sender, multiple receivers; see, e.g., Wang~\cite{wang2013bayesian} and Arieli and Babichenko~\cite{arieli2019private}), or a star network with incoming edges to {the central node} (multiple senders, single receiver; e.g., Gentzkow and Kamenica~\cite{gentzkow2016competition}).  However, {little is known about persuasion} over networks with a diameter greater than one. Our {paper}  solves the Bayesian persuasion problem {on the line graph}. Despite the simplicity of this graph, it naturally captures the communication structure in many realistic examples including those mentioned above.
Persuasion problems on such general networks are unlikely to admit an explicit solution. An example from Appendix~\ref{sect_two_direct_mediators} shows that intuitions developed for the line graph fail even for simple trees. For general directed networks 
even basic questions about feasible distributions of beliefs respecting the network hierarchy become involved, as indicated by Brooks et al.~\cite{brooks2019information}. Exceptions admitting tractable characterizations of feasibility are studied by Galperti and Perego~\cite{galperti2019belief} and Babichenko~et~al.~\cite{babichenko2021multi}.

{In undirected networks, where intermediaries broadcast information to all neighbors and information goes back and forth, the complexity of the network helps the sender. Laclau et al.~\cite{laclau2020robust} show that, if there are two disjoint paths between the sender and the receiver, the sender can implement the outcome of direct persuasion. Since the communication is two-way, the sender can check whether the information has been garbled along one of the paths and resend it using the other one.} 

{Persuasion over communication networks (information flows through a sequence of strategic intermediaries) is not to be confused with direct persuasion of receivers, who reside in a network and are either subject to network externalities as in (Candogan~\cite{candogan2019persuasion},  Candogan and Drakopoulos~\cite{candogan2020optimal}) or have access to neighbors' signals (Egorov and Sonin~\cite{egorov2020persuasion}, Kerman et al.~\cite{kerman2021persuading}). Such models naturally arise in the context of voting or product adoption and, from a technical perspective, are closer to the direct persuasion of one ``aggregated'' receiver than to mediated persuasion.
} 

\paragraph{Constrained persuasion.} Our characterization of the sender's value takes the form of constrained concavification. Constrained concavification with a finite number of linear constraints has appeared in the context of Bayesian persuasion and can be handled via the Lagrangian approach; see, e.g., Le Treust and Tomala~\cite{le2019persuasion}, Doval and Skreta~\cite{doval2018constrained}, and Babichenko et al.~\cite{babichenko2020bayesian}. The type of constraints that 
originate in mediated persuasion is different: there is a continuum of them; they are non-linear and, moreover, non-convex.
As we will see in the examples, the sender may end up maximizing over a set having several connected components. 

\section{Model}
First, we give a high-level description of the model and then discuss it in more detail; the nuanced presentation can be found in Appendix~\ref{app_perfect}.

A sender $S$ communicates with a  receiver~$R$ through a sequence of $n$ mediators $M_1,\ldots,M_n$:
$$    S \to M_1 \to M_2\to\ldots\to M_n \to R.$$
There is a {finite set of states} $\Omega$ and a random state $\omega\in \Omega$ is distributed according to $p\in \Delta(\Omega)$. {The distribution $p$ is the agents' common prior.} The only agent observing the realization of the state is the sender. The sender signals some information about $\omega$ to the first mediator $M_1$, who, in her turn, transmits some information to $M_2$, and so on; the last mediator $M_n$ sends a signal to the receiver. Agents  commit to their signaling policies sequentially  starting from the sender and so the policies selected by predecessors can affect successors' choices.  The commitment is public; hence, each agent, including the receiver, knows  how to interpret the signal that she observes.
Once the receiver gets the signal, she selects an action  $a$ from her set of actions $A$. The chosen action $a$ and the realized state $\omega$ determine the payoffs to all the agents. We denote their utility functions by  $u_S$, $u_{M_1},\ldots, u_{M_n}$, $u_R\,:\,\,A\times \Omega\rightarrow \mathbb{R}$. 
Technical assumptions are  imposed on indirect utilities and are discussed below together with other details of the model.
\medskip

{Consider an agent observing a signal from a certain set of signals $S_{\mathrm{in}}$. Her signaling policy is a map $f:\,S_{\mathrm{in}}\to\Delta(S_{\mathrm{out}}) $ assigning a distribution of a signal to be sent as a function of the observed one.
By definition, a signaling policy can only garble the information contained in the observed signal. The agent is free to choose both $S_{\mathrm{out}}$ and $f$, i.e., what to tell and how, but cannot affect the observed signals $S_{\mathrm{in}}$ and treats this set as given.
We assume that the sets of signals are measurable spaces without specifying their exact nature; e.g., these can be binary messages (the project is worth implementing: yes/no), a finite or countably infinite collection of verbal descriptions, a collection of real numbers (costs and benefits  of implementing the project), or a graph of a function (estimated demand for the final product as a function of its price).
\smallskip

 The sender observes the realization of $\omega$ and so the set $\Omega$ plays the role of her set of observed signals.
Consequently, the sender selects a set of signals $S_S$ and a  map $f_S:\,\Omega\rightarrow\Delta(S_S)$.  Mediator $M_1$ learns the sender's policy choice and selects a set of signals $S_{M_1}$ and a map $f_{M_1}:\, S_S\to \Delta(S_{M_1})$. Inductively, for $i\leq n$, mediator $M_i$ selects $S_{M_i}$ and $f_{M_i}:\, S_{M_{i-1}}\to \Delta(S_{M_{i}})$ depending on the policies $f_S,f_{M_1},\ldots,f_{M_{i-1}}$ of the predecessors. Finally, the receiver, who is aware of all these policy choices, selects a function $f_R:\, S_{M_n}\to A$ specifying how her action depends on the signal sent by the last mediator. 

Once all the agents have decided on their policies $f_S,f_{M_1},\ldots, f_{M_n},f_R$, the state $\omega$ is realized and the signals $s_S, s_{M_1},\ldots, s_{M_n}$ together with the action $a$ are generated sequentially. In other words,
{a profile} of policies combined with the prior $p$ induce the joint distribution of the state $\omega$, the signals $s_S,s_{M_1},\ldots,s_{M_n}$, and the action $a$. The resulting payoffs to  the sender, mediators, and the receiver are given by the expected values of their utilities $u_S(a,\omega)$, $u_{M_1}(a,\omega),\ldots, u_{M_n}(a,\omega)$, and $u_R(a,\omega)$.
\medskip

The above description defines an $(n+2)$-player game. We are interested in its subgame perfect equilibria, which are formally defined in  Appendix~\ref{app_perfect}. The definition is standard apart from a refinement needed to handle degenerate problems with $n\geq 2$ mediators: we assume that no mediator garbles information received from predecessors unless garbling leads to a strictly higher payoff. For example, under this refinement, a mediator whose utility does not depend on the receiver's action never affects the flow of information and can be eliminated.} 

Our goal is to determine the optimal expected payoff that the sender can achieve in a subgame perfect equilibrium. 
{As we will see, {under a mild technical assumption,} the problem is well-defined since the set of subgame perfect equilibria is non-empty and, moreover, there exists a sender's optimal equilibrium.

\subsection{Indirect utilities and technical assumptions}\label{sec_assumpt} As in the standard model of Bayesian persuasion  without mediators \cite{kamenica2011bayesian}, the receiver's posterior belief   incorporates all the information from her signal that is relevant for the action choice. This observation simplifies the analysis: it is enough to keep track of induced beliefs only and  represent the problem via indirect utility functions expressing agents' payoffs as functions of the receiver's belief.

For any belief $q\in\Delta(\Omega)$, the indirect utility of the sender is defined by  $$v_S(q)=\sum_{k\in \Omega} q_k\cdot u_S(a(q),k),\quad\mbox{where}\quad  a(q)\in\argmax_{a\in A} \sum_{k\in \Omega}q_k\cdot u_R(a,k);$$ i.e., $v_S(q)$ is the expected utility of the sender when the receiver picks her best-reply action corresponding to a belief $q$. The indirect utilities $v_{M_1},\ldots,v_{M_n}$ of mediators are defined similarly. If the receiver's best reply is not unique, a particular selection is fixed endogenously.
To avoid technicalities in the body of the paper, we impose the following assumption.
\begin{assumption}\label{asumpt} The receiver has a best-reply selection such that the corresponding  indirect utilities of all the mediators are continuous  in $q\in\Delta(\Omega)$, and the indirect utility of the sender is upper semicontinuous and bounded.
\end{assumption}
The assumption ensures the existence of an equilibrium and allows us to focus on the essence of the problem. On the other hand, this assumption is quite restrictive as it implicitly excludes all the problems where the set of the receiver's actions is discrete and thus the best-reply action cannot change continuously. Assumption~\ref{asumpt} can be dropped at the cost of considering $\varepsilon$-equilibria (see Appendix~\ref{app_perfect}) and replacing maxima by suprema. In Appendices~\ref{app_one_receiver} and~\ref{app_n_mediators}, we formulate and prove extended versions of our results allowing for discontinuities.

Note that Assumption~\ref{asumpt} is satisfied if $A$ is compact, utilities are continuous in $a$, and the receiver's best reply is unique and 
changes continuously in the belief. 
 For instance, $A=\Delta(\Omega)$, the receiver's utility $u_R$ is the quadratic scoring rule, and utilities of all other agents are continuous (the receiver is a market expert aiming to learn the state $\omega$ of a firm based on information released by its PR department, which itself has access to information approved only by the CEO). 
}


\section{One Mediator}\label{sec_one_mediator}
{Consider a problem with one mediator $M_1=M$ and indirect utilities $v_S$ and $v_M$ satisfying Assumption~\ref{asumpt}. The one-mediator case happens to be special as the sender's optimal value admits a simple geometric characterization that does not extend to $n\geq 2$ mediators.}
\ed{This characterization is formulated in terms of affine domination.
\begin{definition}\label{def:see}
 Beliefs $q_1,\ldots,q_{{|\Omega|}} \in \Delta(\Omega)$ are \emph{affine dominating with respect to a function $f:\Delta(\Omega)\to \mathbb{R}$}  if
$$\sum_{k=1}^{{|\Omega|}} \alpha_k\cdot f\big(q_k\big)\geq f\left(\sum_{k=1}^{{|\Omega|}} \alpha_k\cdot q_k\right)$$
for every collection of weights $\alpha=(\alpha_1,\ldots,\alpha_{|\Omega|})\in \Delta(\Omega)$. 
\end{definition}
}

{Let ${D}\subset \Delta(\Omega)^\Omega$ be the set of all {$q_1,\ldots,q_{|\Omega|}$}} that are affine dominating with respect to the indirect utility $v_M$ of the mediator.  
Define the {\emph{constrained concavification}} $\cav_{{D}}[v]:\Delta(\Omega)\rightarrow\mathbb{R}$ of an {upper semicontinuous} function $v:\Delta(\Omega)\rightarrow\mathbb{R} $ 
with respect to a set {${D}\subset \Delta(\Omega)^\Omega$} as follows:  
\begin{align}\label{eq:max}
\cav_{{D}}\big[v\big](p)=\max \left\{ \sum_{k=1}^{{|\Omega|}} \alpha_k \cdot v(q_k) \ \Big| \ 
(q_1,\ldots,q_{{|\Omega|}})\in {D}, \ \alpha\in \Delta(\Omega),\ 
\sum_{k=1}^{{|\Omega|}}\alpha_k q_k=p \right\}.
\end{align}
The maximization is over a non-empty set since  ${D}$ contains  $(p,\ldots,p)$.
As we see, in the constrained concavification, only convex combinations of points from ${D}$ are allowed, whereas in the standard concavification, the maximum is taken over all  $(q_1,\ldots,q_{{|\Omega|}})\in \Delta(\Omega)^\Omega$. 
\begin{theorem}\label{theorem:main}
For any prior $p \in\Delta(\Omega),$ the {sender's optimal payoff is equal to} $\cav_{{D}}\big[v_S\big](p)$.
\end{theorem}

A subgame perfect equilibrium where the sender's payoff is equal to $\cav_{{D}}\big[v_S\big](p)$ can be constructed as follows. 
The sender's policy induces the optimal mediator's beliefs $q_1,\ldots, q_{|\Omega|}$ from~\eqref{eq:max} with the respective probabilities $\alpha_1,\ldots, \alpha_{|\Omega|}$. \ed{By the result of Lipnowski~et al.~\cite[Lemma 1]{lipnowski2020attention}, the affine-domination property implies that the mediator has no incentive to garble the information and so she reveals it fully to the receiver as one of her best replies.}
\begin{corollary}\label{cor:k-signals}
There exists a sender's optimal subgame perfect equilibrium in which the sender uses at most {$|\Omega|$} different signals to persuade the mediator and the mediator fully reveals the sender's signal to the receiver, i.e., {the mediator does not garble the information on the equilibrium path.}
\end{corollary}
\ed{We will see, for $n\geq 2$ mediators, $|\Omega|$ signals are no longer enough there. It is instructive to compare  Corollary~\ref{cor:k-signals} with the result of  Arieli et al.~\cite{arieli2021feasible} who show that, for two receivers, a binary state, and no mediators, optimal persuasion may require an infinite number of signals.

The equilibrium described above ensures that the sender's optimal payoff is at least $\cav_{{D}}\big[v_S\big](p)$. 
A high-level intuition behind the the upper bound is as follows. By a version of the revelation principle, it is enough to maximize sender's payoff over equilibria in which
the whole garbling is done by the sender and the mediator transmits the information without garbling it. The requirement that the mediator has no incentive to garble boils down to the property of affine domination. For continuous $v_S$ and $v_M$ this intuition can be converted to a formal proof using the technique of Lipnowski~et al.~\cite{lipnowski2020attention}.

 As we mentioned in Section~\ref{sec_assumpt}, discontinuous utilities naturally arise.  Handling discontinuous $v_S$ and $v_M$
 complicates the proof of Theorem~\ref{theorem:main} as equilibria are not guaranteed to exist and we need to deal with $\varepsilon$-equilibria. A version of 
 Theorem~\ref{theorem:main} allowing for discontinuities is formulated and proved  in Appendix~\ref{app_one_receiver}.
 }

{Theorem~\ref{theorem:main} also implies the unconstrained concavification formula of Kamenica and Gentzkow~\cite{kamenica2011bayesian} for the   non-mediated persuasion. This classic setting can be emulated by considering a dummy mediator whose indirect utility is constant.\footnote{We rely on our refinement of subgame perfection that eliminates equilibria where the dummy mediator does not pass on any information to the receiver. In the one-mediator case such equilibria can alternatively be eliminated by the fact that we are interested in the sender's optimal equilibria.} For such a dummy mediator, ${D}=\Delta(\Omega)^\Omega,$ and we obtain that the sender's optimal payoff is equal to the unconstrained concavification $\cav[v_S]=\cav_{\Delta(\Omega)^\Omega}[v_S]$. Thus, adding a non-dummy mediator (${D}\ne\Delta(\Omega)^\Omega$) can be seen as adding constraints to the previously unconstrained concavification.}
\begin{corollary}\label{cor_decrease}
{The sender's optimal value can only decrease after adding a mediator.}
\end{corollary}
\ed{In Section~\ref{sec:more}, we will show that adding one more mediator can benefit the sender.}

\begin{example}
{Consider binary $\Omega=\{0,1\}$ and the sender with indirect utility {depicted} in Figure \ref{fig:sender} (we identify $\Delta (\{0,1\})$ with $[0,1]$).}
  
\begin{figure}[H]
    \centering
    \caption{The indirect utility {$v_S$} of the sender.}
    \label{fig:sender}

\begin{tikzpicture}[line width=2pt, scale=0.8]
    \pgfplotsset{every tick/.append style={semithick, major tick length=5pt, minor tick length=5pt, black}};
    \begin{axis}[
    axis line style = ultra thick,
    axis x line=center,
    axis y line=none,
    no marks]
    \addplot+[smooth,very thick, blue,name path=A,domain=0:0.2] {4-100*(0.2-x)*(0.2-x)};
    \addplot+[smooth,very thick,blue,name path=B,domain=0.2:0.8] {3+cos(deg(10*pi*(x-0.2)/3))};
    \addplot+[smooth,very thick,blue,name path=c,domain=0.8:1] {4-100*(x-0.8)*(x-0.8)};
  \end{axis}
\end{tikzpicture}
\end{figure}
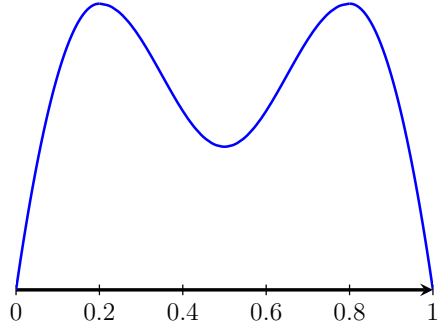

{The unconstrained concavification $\cav[v_S]$ (the sender's optimal value without a mediator) can be constructed as follows.} We connect all pairs of points on the graph of {$v_S$} {by linear segments.} {This results in the shaded region from  Figure \ref{fig:cav}.} The upper {boundary} of this {region is the graph of}~{$\cav[v_S]$.}

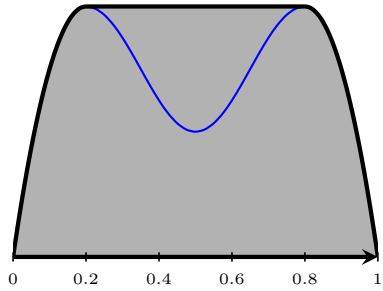
\begin{figure}[H]
    \centering
    \caption{The {region obtained} by connecting all pairs of points {on the graph of $v_S$.}}
    \label{fig:cav}

\begin{tikzpicture}[scale=0.7]
    \pgfplotsset{every tick label/.append style={font=\tiny, very thick}};
    \pgfplotsset{every tick/.append style={semithick, major tick length=5pt, minor tick length=5pt, black}};
    \begin{axis}[axis on top=true,
     axis line style = ultra thick,
     axis x line=center,
    axis y line=none,
    no marks]
    \addplot+[smooth,ultra thick,black,name path=A,domain=0:0.2] {4-100*(0.2-x)*(0.2-x)};
    \addplot+[smooth, thick,  blue,      name path=B,domain=0.2:0.8] {3+cos(deg(10*pi*(x-0.2)/3))};
    \addplot+[smooth,ultra thick,black,name path=C,domain=0.8:1] {4-100*(x-0.8)*(x-0.8)};
    \addplot+[draw=none,         name path=D,domain=0:1] {0};
    \addplot+[ultra thick,black,       name path=E,domain=0.2:0.8] {4};
    
    \addplot+[gray!60] fill between[of=A and D];
    \addplot+[gray!60] fill between[of=E and D];
    \addplot+[gray!60] fill between[of=C and D];
  \end{axis}
\end{tikzpicture}
\end{figure}

{Let us introduce a mediator with the utility function $v_M$ presented in Figure \ref{fig:u3}. }

%
%
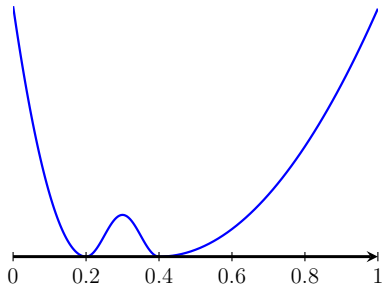
\begin{figure}[H]
    \centering
    \caption{The utility {$v_M$} of the  mediator.}
    \label{fig:u3}

\begin{tikzpicture}[scale=0.7]
\pgfplotsset{every tick/.append style={semithick, major tick length=5pt, minor tick length=5pt, black}};
    \begin{axis}[
     axis on top=true,
     axis line style = ultra thick,
    axis x line=center,
    axis y line=none,
    no marks]
    \addplot+[smooth, very thick, blue,name path=A,domain=0:0.2] {100*(0.2-x)*(0.2-x)};
    \addplot+[smooth,very thick,blue,name path=B,domain=0.2:0.4] {1/3-cos(deg(10*pi*(x-0.2)))/3};
    \addplot+[smooth,very thick,blue,name path=B,domain=0.4:1] {11*(x-0.4)*(x-0.4)};
  \end{axis}
\end{tikzpicture}
\end{figure}
\smallskip

{By Theorem~\ref{theorem:main}, we can find the sender's value as we did in the no-mediator case}, but instead of connecting \emph{all} pairs of points on {the graph of $v_S$}, we {need to} connect only {those pairs that are} affine dominating {with respect to $v_M$.} As a warm-up, assume that the prior is $p=1/2$. The optimal value {without a mediator} is obtained by connecting the pair of points on {the graph of $v_S$} with the beliefs $q_1=0.2$ and $q_2=0.8$. Let us check whether this pair of posteriors is affine dominating with respect to {$v_M$.}

%
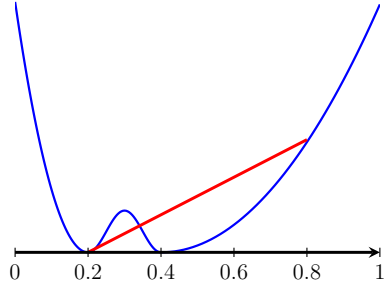
\begin{figure}[H]
    \centering
    \caption{The pair of posteriors $q_1=0.2$ and $q_2=0.8$ on the graph of {$v_M$.}}
    \label{fig:see13}

\begin{tikzpicture}[scale=0.7]
\pgfplotsset{every tick/.append style={semithick, major tick length=5pt, minor tick length=5pt, black}};
    \begin{axis}[
     axis on top=true,
     axis line style = ultra thick,
    axis x line=center,
    axis y line=none,
    no marks]
    \addplot+[smooth,very thick,blue,name path=A,domain=0:0.2] {100*(0.2-x)*(0.2-x)};
    \addplot+[smooth,very thick,blue,name path=B,domain=0.2:0.4] {1/3-cos(deg(10*pi*(x-0.2)))/3};
    \addplot+[smooth,very thick,blue,name path=B,domain=0.4:1] {11*(x-0.4)*(x-0.4)};
    \addplot+[ultra thick,red,domain=0.2:0.8] {3*(x-0.2)};
  \end{axis}
\end{tikzpicture}
\end{figure}
\smallskip

As we see  in Figure \ref{fig:see13}, the posteriors $q_1=0.2 $ {and} $q_2=0.8$ are not affine dominating with respect to {$v_M$} {since} there {are} points {on the graph that lie above the segment.} {We conclude that the optimal utility achieved by the sender without a mediator for $p=1/2$ can no longer be achieved in the presence of the mediator having the utility function $v_M$.} 

{To get more intuition about affine domination,} we {plot} the \emph{set} of points {$q_2$} {such} that a given point {$q_1$} paired with them forms an affine dominating pair with respect to $v_M$. {Figure~\ref{fig:see23} demonstrates this set for ${q_1}=0.15$  (the linear segments between the pairs are traced to visualize the construction).}

\begin{figure}[h]
    \centering
    \caption{{The set of $q_2$ such that the pair $q_1,q_2$ with $q_1=0.15$ is} {affine dominating} with respect to {$v_M$.}}
    \label{fig:see23}

\begin{tikzpicture}[scale=0.8]
\pgfplotsset{every tick/.append style={semithick, major tick length=5pt, minor tick length=5pt, black}};
    \pgfplotsset{every tick label/.append style={font=\tiny}};
    \begin{axis}[
     axis on top=true,
     axis line style = ultra thick,
    axis x line=center,
    xtick={0, 0.15, 0.29, 0.87, 1},
    axis y line=none,
    no marks]
    \addplot+[smooth,very thick,blue,name path=A,domain=0:0.2] {100*(0.2-x)*(0.2-x)};
    \addplot+[smooth,very thick,blue,name path=B,domain=0.2:0.29] {1/3-cos(deg(10*pi*(x-0.2)))/3};
    \addplot+[smooth,very thick,blue,name path=BB,domain=0.29:0.4] {1/3-cos(deg(10*pi*(x-0.2)))/3};
    \addplot+[smooth,very thick,blue,name path=C,domain=0.4:1] {11*(x-0.4)*(x-0.4)};
    \addplot+[draw=none, name path=D, domain=0.15:0.2] {100*0.05*0.05+3*(x-0.15)};
    \addplot+[draw=none, name path=DD, domain=0.2:0.29] {100*0.05*0.05+3*(x-0.15)};
    \addplot+[draw=none, name path=DDD, domain=0.29:0.87] {100*0.05*0.05+3*(x-0.15)};
    \addplot+[draw=none, name path=E, domain=0.15:1] {100*0.05*0.05+4.35*(x-0.15)};
    \addplot+[draw=none, name path=F, domain=0:0.15] {100*0.05*0.05-25*(x-0.15)};
    \addplot+[gray!60] fill between[of=A and F];
    \addplot+[gray!60] fill between[of=A and D];
    \addplot+[gray!60] fill between[of=B and DD];
    \addplot+[gray!60] fill between[of=E and DDD];
  \end{axis}
\end{tikzpicture}
\end{figure}
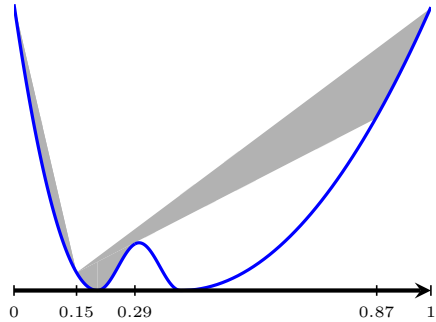
\smallskip

{From Figure \ref{fig:see23},} we deduce that pairing the point $q_1=0.15$ with all the points ${q_2}\in[0,0.29]\cup [0.87,1]$ forms  affine dominating pairs {with respect to~$v_M$.}

{By Theorem~\ref{theorem:main}, to find the sender's optimal value, we need to connect all the pairs of points $q_1,q_2$ on the graph of $v_S$
that are affine dominating with respect to $v_M$; the graph of the value as a function of the prior will then be given by the upper boundary of the resulting region. For example, the point $q_1=0.15$ is to be connected to $q_2\in[0,0.29]\cup [0.87,1]$; see Figure \ref{fig:see3}. By connecting all the affine dominating pairs $q_1,q_2$, we obtain the region depicted in Figure~\ref{fig:sol}.}


\begin{figure}[h]
    \centering
    \caption{{The region obtained by connecting all the pairs of points $q_1,q_2$ that are affine dominating with respect to the mediator's utility with $q_1=0.15$ on the graph of the sender's utility.}}
    \label{fig:see3}

\begin{tikzpicture}[scale=0.8]
\pgfplotsset{every tick/.append style={semithick, major tick length=5pt, minor tick length=5pt, black}};
    \pgfplotsset{every tick label/.append style={font=\tiny}};
    \begin{axis}[
      axis on top=true,
     axis line style = ultra thick,
    axis x line=center,
    xtick={0, 0.15, 0.29, 0.87, 1},
    axis y line=none,
    no marks]
    \addplot+[smooth,very thick, blue,name path=A,domain=0:0.2] {4-100*(0.2-x)*(0.2-x)};
    \addplot+[smooth,very thick,blue,name path=B,domain=0.2:0.29] {3+cos(deg(10*pi*(x-0.2)/3))};
    \addplot+[smooth,very thick,blue,name path=BB,domain=0.29:0.8] {3+cos(deg(10*pi*(x-0.2)/3))};
    \addplot+[smooth,very thick,blue,name path=C,domain=0.8:0.87] {4-100*(x-0.8)*(x-0.8)};
    \addplot+[smooth,very thick,blue,name path=CC,domain=0.87:1] {4-100*(x-0.8)*(x-0.8)};
    \addplot+[draw=none, name path=D,domain=0:0.15] {4-100*0.05*0.05+25*(x-0.15)};
    \addplot+[draw=none, name path=E,domain=0.15:0.29] {4-100*0.05*0.05-1.2*(x-0.15)};
    \addplot+[draw=none, name path=F,domain=0.15:0.87] {4-100*0.05*0.05-0.5*(x-0.15)};
    \addplot+[draw=none, name path=G,domain=0.15:1] {4-100*0.05*0.05-4.4*(x-0.15)};
    
    \addplot+[gray!60] fill between[of=A and D];
    \addplot+[gray!60] fill between[of=A and E];
    \addplot+[gray!60] fill between[of=B and E];
    \addplot+[gray!60] fill between[of=F and G];
    \addplot+[gray!60] fill between[of=CC and G];
  \end{axis}
\end{tikzpicture}
\end{figure}
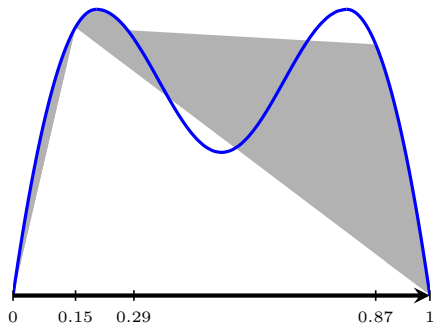


\begin{minipage}[b]{0,45\textwidth}
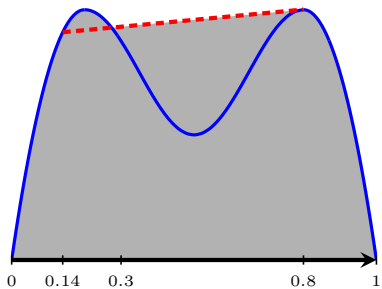
\begin{figure}[H]
    \centering
    \caption{The {region} obtained by connecting all pairs of points {that are affine dominating with respect to the mediator's utility.} {The endpoints of the dashed red segment form the pair of posterior beliefs used by the sender to achieve the corresponding payoff.}}
    \label{fig:sol}
\ifdefined\AEJ
\vskip 0.7cm
\fi
\begin{tikzpicture}[scale=0.7]
\pgfplotsset{every tick/.append style={semithick, major tick length=5pt, minor tick length=5pt, black}};
    \pgfplotsset{every tick label/.append style={font=\tiny}};
    \begin{axis}[
     axis on top=true,
     axis line style = ultra thick,
    axis x line=center,
    xtick={0, 0.14, 0.3, 0.8,1},
    axis y line=none,
    no marks]
    \addplot+[smooth,very thick,blue,name path=A,domain=0:0.2] {4-100*(0.2-x)*(0.2-x)};
    \addplot+[smooth,very thick,blue,name path=B,domain=0.2:0.29] {3+cos(deg(10*pi*(x-0.2)/3))};
    \addplot+[smooth,very thick,blue,name path=BB,domain=0.29:0.8] {3+cos(deg(10*pi*(x-0.2)/3))};
    \addplot+[smooth,very thick,blue,name path=C,domain=0.8:0.87] {4-100*(x-0.8)*(x-0.8)};
    \addplot+[smooth,very thick,blue,name path=CC,domain=0.87:1] {4-100*(x-0.8)*(x-0.8)};
    \addplot+[draw=none, name path=D,domain=0:0.15] {4-100*0.05*0.05+25*(x-0.15)};
    \addplot+[draw=none, name path=E,domain=0.15:0.29] {4-100*0.05*0.05-1.2*(x-0.15)};
    \addplot+[draw=none, name path=F,domain=0.15:0.87] {4-100*0.05*0.05-0.5*(x-0.15)};
    \addplot+[draw=none, name path=G,domain=0.15:1] {4-100*0.05*0.05-4.4*(x-0.15)};
    \addplot+[draw=none, name path=H,domain=0.3:0.79] {3+cos(deg(10*pi*0.1/3))+(x-0.3)};
    \addplot+[draw=none, name path=I,domain=0:1] {0};
    
    \addplot+[gray!60] fill between[of=A and D];
    \addplot+[gray!60] fill between[of=A and E];
    \addplot+[gray!60] fill between[of=B and E];
    \addplot+[gray!60] fill between[of=F and G];
    \addplot+[gray!60] fill between[of=CC and G];
    \addplot+[gray!60] fill between[of=H and I];
    \addplot+[gray!60] fill between[of=C and I];
    \addplot+[gray!60] fill between[of=A and I];
    
    \addplot+[ultra thick, red, name path=J,domain=0.14:0.8]
    {4-100*0.06*0.06+0.55*(x-0.14)};
    
    \addplot+[gray!60] fill between[of=J and I];
    
    
  \end{axis}
\end{tikzpicture}
\end{figure}
\end{minipage}
\hskip 1cm
\begin{minipage}[b]{0,45\textwidth}

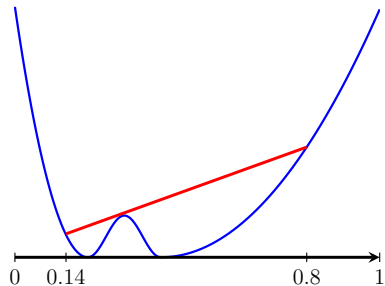
\begin{figure}[H]
    \centering
    \caption{{The pair of beliefs corresponding to the dashed red segment from Figure~\ref{fig:sol} is ``extreme'' with respect to $v_M$; e.g., one cannot increase $0.14$ to a higher belief (which would also increase the sender's utility) without violating the requirement of affine domination.}}
    \label{fig:sol2}
    \ifdefined\AEJ
    \vskip 0.071cm
    \fi
\begin{tikzpicture}[scale=0.7]
\pgfplotsset{every tick/.append style={semithick, major tick length=5pt, minor tick length=5pt, black}};
    \begin{axis}[
     axis on top=true,
     axis line style = ultra thick,
    axis x line=center,
    xtick={0, 0.14, 0.8, 1},
    axis y line=none,
    no marks]
    \addplot+[smooth,very thick,blue,name path=A,domain=0:0.2] {100*(0.2-x)*(0.2-x)};
    \addplot+[smooth,very thick,blue,name path=B,domain=0.2:0.4] {1/3-cos(deg(10*pi*(x-0.2)))/3};
    \addplot+[smooth,very thick,blue,name path=B,domain=0.4:1] {11*(x-0.4)*(x-0.4)};
    \addplot+[red,ultra thick,domain=0.14:0.8] {11*0.4*0.4+2.1*(x-0.8)};
  \end{axis}
\end{tikzpicture}

\end{figure}
\end{minipage}
\smallskip

Similarly to unconstrained concavification {without a mediator,} the {constrained} concavification {depicted in Figure \ref{fig:sol}} tells us the structure of the {sender's} optimal policy. For priors {$p \leq 0.3$}, {it is optimal to reveal no information.} For {$0.3 \leq p \leq 0.8$,} the optimal policy 
{induces the pair of posterior beliefs $q_1=0.14$ and $q_2=0.8$ of the mediator, which correspond to the red linear segments in Figures~\ref{fig:sol} and~\ref{fig:sol2}.} For {$0.8 \leq p$,} the no-information {policy} becomes optimal again.  

\end{example}

\section{More Than One Mediator}\label{sec:more}

The case of $n\geq 2$ mediators turns out to be substantially different from that with one mediator because of the interplay between the mediators' incentives. 

For several mediators, one might conjecture a natural extension of Theorem \ref{theorem:main} by requiring the constrained concavification to satisfy affine domination with respect to the indirect utilities $v_{M_i}$ of all mediators  $M_i$, $i=1,\ldots, n$.
Indeed, this requirement incentivizes all the mediators to reveal information fully to their successors. Despite this intuition, the conjecture turns out to be wrong. 
Such a constrained concavification gives a lower bound on the sender's optimal value, which can be proved via arguments similar to those from the proof of Theorem~\ref{theorem:main}. However, there are examples where the sender can achieve a higher value by exploiting the cancellation of the mediators' incentives. 

Due to the cancellation of incentives, adding the second mediator to a one-mediator problem can benefit the sender. This phenomenon is surprising when compared with Corollary~\ref{cor_decrease} which claims that the first mediator is never profitable. We illustrate the cancellation of incentives and the profitability of the second mediator in the following example.
\begin{example}[Interplay of mediators' incentives matters]\label{ex_interplay}
{Consider a problem with two mediators} and indirect utilities $v_S,v_{M_1},v_{M_2}$  depicted in Figure \ref{fig:util}, where $\varepsilon>0$ is a small fixed parameter.
\begin{figure}[h]
    \centering
    \caption{The utilities of the sender ($v_{S}$) and the mediators ($v_{M_1}$ and $v_{M_2}$).}
    \label{fig:util}
    \vskip 0.3cm
    \begin{tikzpicture}[scale=0.55]
\pgfplotsset{every tick/.append style={semithick, major tick length=5pt, minor tick length=5pt, black}};
\pgfplotsset{every tick label/.append style={font=\Large}};
    \begin{axis}[
    axis on top=true,
    axis line style = ultra thick,
    axis x line=center,
    xtick={1},
    ytick={1},
    axis y line=center,
    ylabel=$v_S$,
    y label style={at={(-0.25,1.1)}, font=\huge},
    no marks]
    \addplot+[smooth,ultra thick,blue,name path=A,domain=0:0.1] {1-10*x};
    \addplot+[smooth,ultra thick,blue,name path=A,domain=0.01:0.9] {0.005};
    \addplot+[smooth,ultra thick,blue,name path=A,domain=0.9:1] {1+10*(x-1)};

    \end{axis}
    \node[] at (0.3,0.3) {$\varepsilon$};
    \node at (6.5,0.3) {$\varepsilon$};
\end{tikzpicture}
\begin{tikzpicture}[scale=0.55]
\pgfplotsset{every tick/.append style={semithick, major tick length=5pt, minor tick length=5pt, black}};
\pgfplotsset{every tick label/.append style={font=\Large}};
    \begin{axis}[
    axis on top=true,
    axis line style = ultra thick,
    axis x line=center,
    xtick={0.25,1},
    ytick={0.5,1},
    axis y line=center,
    ylabel=$v_{M_1}$,
    y label style={at={(-0.35,1.1)}, font=\huge},
    no marks]
    \addplot[ultra thick,blue,name path=A,domain=0:0.1] {0.5-5*x};
    \addplot[ultra thick,blue,name path=A,domain=0.1:0.2] {0.005};
    \addplot[ultra thick,blue,name path=A,domain=0.2:0.25] {20*(x-0.2)};
    \addplot[ultra thick,blue,name path=A,domain=0.25:0.3] {1-20*(x-0.25)};
    \addplot[ultra thick,blue,name path=A,domain=0.3:0.9] {0.005};
    \addplot[ultra thick,blue,name path=A,domain=0.9:1] {0.5+5*(x-1)};

    \end{axis}
    \node[] at (0.2,0.3) {$\varepsilon$};
    \node at (6.5,0.3) {$\varepsilon$};
    \node at (1.7,0.3) {$\varepsilon$};
\end{tikzpicture}
\begin{tikzpicture}[scale=0.55]
\pgfplotsset{every tick/.append style={semithick, major tick length=5pt, minor tick length=5pt, black}};
\pgfplotsset{every tick label/.append style={font=\Large}};
    \begin{axis}[
    axis on top=true,
    axis line style = ultra thick,
    axis x line=center,
    xtick={0.5,1},
    ytick={0.8,1},
    axis y line=center,
    ylabel=$v_{M_2}$,
    y label style={at={(-0.35,1.1)}, font=\huge},
    no marks]
    \addplot[smooth,ultra thick,blue,name path=A,domain=0:0.1] {1-10*x};
    \addplot[smooth,ultra thick,blue,name path=A,domain=0.1:0.45] {0.01};
    \addplot[smooth,ultra thick,blue,name path=A,domain=0.45:0.5] {16*(x-0.45)};
    \addplot[smooth,ultra thick,blue,name path=A,domain=0.5:0.55] {16*(0.55-x)};
    \addplot[smooth,ultra thick,blue,name path=A,domain=0.55:0.9] {0.01};
    \addplot[smooth,ultra thick,blue,name path=A,domain=0.9:1] {1+10*(x-1)};

    \end{axis}
    \node[] at (0.3,0.3) {$\varepsilon$};
    \node at (6.5,0.3) {$\varepsilon$};
    \node at (3.4,0.3) {$\varepsilon$};
\end{tikzpicture}

\end{figure}
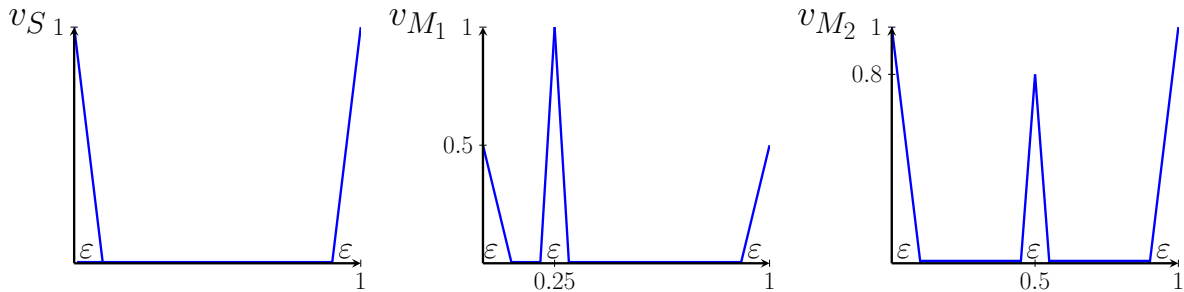

Assume that the prior is $p=\frac{1}{2}$.
Note that the best outcome for the sender is obtained if the state is fully revealed to the receiver; i.e., her posterior beliefs are $0$ or $1$. Can the sender achieve this outcome in an equilibrium? Note that with respect to mediator $M_1$'s utility, the pair of posteriors $0$ and $1$ is {not} affine dominating. Hence, in the absence of mediator $M_2$, the first mediator would prefer to garble the information and place the receiver's posteriors at $0.25$ and $1$ with the corresponding probabilities $\frac{2}{3}$ and $\frac{1}{3}$. In the presence of $M_2$,  such a garbling is no longer profitable since $M_2$ gets zero utility from the posterior $0.25$ and so she has an incentive to garble the information even more. 

Let us check that, in the presence of $M_2$,  the full revelation is a best reply of $M_1$ to the full-revelation policy of the sender. Consider the subgame starting from mediator $M_1$. Since $M_1$ is aware of the realized state, we can treat her as a sender in this subgame; i.e., we are back to the case of a single mediator and can use Theorem~\ref{theorem:main} to determine the optimal payoff to $M_1$.
The set of pairs $q_1<\frac{1}{2}<q_2$ that are affine dominating with respect to $v_{M_2}$ satisfy $q_1\in[0,\,\varepsilon]$ and $q_2\in[1-\varepsilon,\,1].$ The optimal payoff to $M_1$ is achieved at $q_1=0$ and $q_2=1$, i.e., at  the full-revelation policy. Hence, the full-revelation  policy is a best reply of $M_1$ to the full-revelation policy of $S$. For mediator $M_2$, the full revelation is also a best reply as it results in her ideal utility. Thus the full-revelation policies of all the agents form the equilibrium path of a subgame perfect equilibrium.

We conclude that the presence of $M_2$ eliminates some of the profitable deviations of $M_1$ and allows the sender to achieve her ideal utility of~$1$.

\end{example}

\subsection{Characterization of the value for a sequence of mediators}
Example~\ref{ex_interplay} demonstrates that in a problem with several mediators, the sender may be able to achieve a strictly higher payoff than prescribed by the constrained concavification that does not take into account the interplay of mediators' incentives. 

Although the naive concavification approach does not work, we characterize  the  sender's  optimal value and the characterization takes the form of a constrained concavification. To capture the interplay of incentives, the constraints restrict distributions of beliefs, in contrast to the one-mediator case where the constraints restrict the beliefs themselves. 

To formulate the constraints, we  need the following notation. For a pair of belief distributions $\mu,\nu\in \Delta(\Delta(\Omega))$, we write $\mu\succeq \nu$ if $\mu$ is  a mean-preserving spread\footnote{A distribution $\mu$ is a mean-preserving spread of $\nu$ if there is a pair of random variables $q_1$ with  distribution $\nu$ and $q_2$ with distribution $\mu$ defined on the same probability space and forming a martingale, i.e., such that $\mathbb{E}[q_2\mid q_1]=q_1$.} of $\nu$; equivalently, $\nu$ is a mean-preserving contraction of $\mu$. The importance of this notion comes from the fact that, if a signal $s_{\mathrm{in}}$ induces a belief distribution $\mu$ and its garbling $s_{\mathrm{out}}$ obtained via some signaling policy induces a belief distribution $\nu$, then we have $\mu\succeq \nu$. Moreover, if $s_{\mathrm{in}}$ induces the belief distribution $\mu$ and $\nu\preceq \mu$, then one can find a signaling policy that induces $\nu$; see Appendix~\ref{subsec_Blackwell}.

Consider a sender communicating with a receiver through a sequence of $n$ mediators. We assume that the indirect utility $v_{M_i}$ of each mediator $M_i$ is continuous and the indirect utility of sender $v_S$ is a bounded upper semicontinuous function on $\Delta(\Omega)$; i.e., Assumption~\ref{asumpt} holds.

For $\mu\in \Delta(\Delta(\Omega))$ and a function $f$ on $\Delta(\Omega)$, the expected value of $f(q)$ with $q$ distributed according to $\mu$ is denoted by $\mathbb{E}_\mu[f(q)]=\int_{\Delta(\Omega)} f(q)\dd\mu(q)$. Define sets ${\mathcal{M}}_i\subset \Delta(\Delta(\Omega))$ for $i=1,\ldots, n+1$ recursively:  ${\mathcal{M}}_{n+1}=\Delta\big(\Delta(\Omega)\big)$ and ${\mathcal{M}}_i$ with $i=1,\ldots, n$ is expressed through ${\mathcal{M}}_{i+1}$:
\begin{align}\label{eq_FRi_definition}
        {\mathcal{M}}_i&=\Big\{\mu\in {\mathcal{M}}_{i+1} \ \Big|\  \Big(\nu\in {\mathcal{M}}_{i+1}, \ \nu\preceq \mu\Big)\Longrightarrow   \mathbb{E}_\mu\big[v_i(q) \big]\geq \mathbb{E}_\nu\big[v_i(q) \big]\Big\}.
\end{align}
The intuition for this definition is that ${\mathcal{M}}_i$ consists of belief distributions that the mediator $M_i$ and all her successors do not have an incentive to garble. For a given $\mu$, such an incentive can be absent for one of the two reasons captured by the recursive formula: either the outcome $\nu$ of a hypothetical garbling gives lower utility $\mathbb{E}_\nu\big[v_i(q) \big] \leq \mathbb{E}_\mu\big[v_i(q) \big]$ to the mediator $M_i$ or $\nu$ will be garbled even further by her successors, i.e., $\nu\notin {\mathcal{M}}_{i+1}$.

The definition of constrained concavification~\eqref{eq:max} extends to distributional constraints as follows. For a closed subset  $\mathcal{M}\subset \Delta(\Delta(\Omega))$ and an upper semicontinuous function $v$ on $\Delta(\Omega)$, we define the constrained concavification of $v$ with respect to  $\mathcal{M}$  by
\begin{equation}\label{eq_constrained_concavification_extended_body}
\cav_{\mathcal{M}}\big[v_S\big](p)=
\max\left\{ \mathbb{E}_\mu [v_S(q)] \  \Big| \ \mu\in\mathcal{M}, \  \mathbb{E}_\mu[q]=p \right\}.
\end{equation}
\begin{theorem}\label{th_n_mediators}
For any number $n$ of mediators and any prior $p \in\Delta(\Omega),$ the {sender's optimal payoff in a subgame  perfect equilibrium is equal to} $\cav_{{\mathcal{M}}_1}\big[v_S\big](p)$, where the set ${\mathcal{M}}_1$ is defined as in~\eqref{eq_FRi_definition}.
\end{theorem}
In Appendix~\ref{app_n_mediators}, we formulate and prove an extended version of this theorem characterizing the sender's optimal value over $\varepsilon$-equilibria for given $\varepsilon>0$ and allowing for discontinuous utilities. 

Under the original continuity assumptions of Theorem~\ref{th_n_mediators}, the constrained concavification $\cav_{{\mathcal{M}}_1}\big[v_S\big](p)$ is well-defined  since the maximization is over a closed non-empty set of distributions. Closedness follows from the continuity assumption on mediators' utilities and non-emptiness follows from the fact that ${\mathcal{M}}_i$ contains the point mass $\mu=\delta_p$ concentrated at $p$. 
An equilibrium,  where the sender's payoff is equal to $\cav_{{\mathcal{M}}_1}\big[v_S\big](p),$ can be constructed explicitly: the sender selects a policy such that the induced  distribution of the first mediator's belief is equal to the optimal $\mu\in {\mathcal{M}}_1$ from~\eqref{eq_constrained_concavification_extended_body}. The recursive definition of ${\mathcal{M}}_i$ ensures that no mediator has an incentive to garble the sender's signal and so the resulting sender's payoff is equal to $\mathbb{E}_\mu[v_s(q)]=\cav_{{\mathcal{M}}_1}\big[v_S\big](p)$.
The main difficulty in the proof is to show that the sender cannot do better: this is proved in Appendix~\ref{app_n_mediators}, where it is also shown that the informal description given above results in an equilibrium.

Abstracting from technical details, the proof of Theorem~\ref{th_n_mediators} relies on two insights: (1) a version of the revelation principle allowing us to focus on those equilibria where only the sender garbles information and (2) a recursive representation of those beliefs that propagate through the sequence of all the mediators without further garbling. It turns out that this high-level reasoning can be used to characterize the optimal payoff of the first player in a broad class of sequential games, where players move a token over a partially ordered set in a monotone way and get payoffs determined by the final position of the token (in our model, the set is the set of belief distributions $\Delta(\Delta(\Omega))$ endowed with the partial order $\succeq$). Appendix~\ref{sect_partial_order} explains the details and gives a unifying perspective on our results and those about sequential persuasion with multiple senders.

Theorem~\ref{theorem:main} is a refinement of Theorem~\ref{th_n_mediators} in the case of one mediator $M_1=M$. Indeed, the set of distributions ${\mathcal{M}}_1$ consists of all  $\mu\in \Delta(\Delta(\Omega))$ such that $\mathbb{E}_\nu[v_M(q)]\leq \mathbb{E}_\mu[v_M(q)]$ for any $\nu \preceq \mu$, while Theorem~\ref{theorem:main} shows that it is enough to concavify over a subset of ${\mathcal{M}}_1$ determined by restrictions  on the support of $\mu$ only. Specifically, ${\mathcal{M}}_1$ can be replaced by the  set of all distributions supported on at most $|\Omega|$ points $(q_1,\ldots,q_{{|\Omega|}})$ that are affine dominating with respect to the mediator's indirect utility $v_M$.

The following example illustrates an application of Theorem~\ref{th_n_mediators}. It demonstrates that the case of several mediators does not admit the simplifications used in Theorem~\ref{theorem:main}. 
We will see that optimal persuasion may require more than $|\Omega|$ signals. Moreover, the set of the receiver's belief distributions that the sender can induce is no longer determined solely by the affine-domination property on the support but also depends on the probabilities with which each belief arises. In other words, the constraints in the constrained concavification do not boil down to constraints on beliefs in $\Delta(\Omega)$ and are unavoidably formulated as restrictions on distributions over beliefs, i.e., on elements of $\Delta(\Delta(\Omega))$.

\begin{example}[$|\Omega|$ signals are not enough for persuasion with $n\geq 2$ mediators]
Consider a problem with two mediators, a binary state $\omega\in\Omega=\{0,1\}$, and indirect utilities as depicted in Figure \ref{fig:util2}. We will see that for the prior $p=0.25$, the sender's optimal policy needs three signals, not two as in the case of one mediator. 
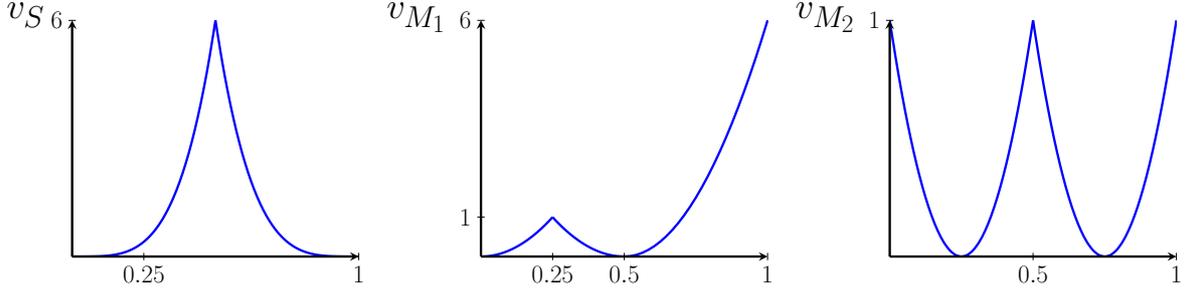
\begin{figure}[h]
    \centering
    \caption{Indirect utilities $v_{S}$, $v_{M_1}$, and $v_{M_2}$  of the sender and the mediators.}
    \label{fig:util2}
    \vskip 0.3cm
    \begin{tikzpicture}[scale=0.55]
\pgfplotsset{every tick/.append style={semithick, major tick length=5pt, minor tick length=5pt, black}};
\pgfplotsset{every tick label/.append style={font=\Large}};
    \begin{axis}[
    axis on top=true,
    axis line style = ultra thick,
    axis x line=center,
    xtick={0.25, 1}, ytick={0,6},
    axis y line=center,
    ylabel=$v_S$,
    y label style={at={(-0.25,1.1)}, font=\huge},
    no marks]
    \addplot[smooth,ultra thick,blue,name path=A,domain=0:0.5] {96*x^4};
    \addplot[smooth,ultra thick,blue,name path=A,domain=0.5:1] {96*(1-x)^4};
    \draw [dashed,gray] (0.5,0) -- (0.5,1);

    \end{axis}
\end{tikzpicture}
\begin{tikzpicture}[scale=0.55]
\pgfplotsset{every tick/.append style={semithick, major tick length=5pt, minor tick length=5pt, black}};
\pgfplotsset{every tick label/.append style={font=\Large}};
    \begin{axis}[
    axis on top=true,
    axis line style = ultra thick,
    axis x line=center,
    xtick={0.25,0.5, 1}, ytick={0,1,6},
    axis y line=center,
    ylabel=${v_{M_1}}$,
    y label style={at={(-0.35,1.1)}, font=\huge},
    no marks]
    \addplot[ultra thick,blue,name path=A,domain=0.0:0.25] {16*(x)^2};
    \addplot[ultra thick,blue,name path=A,domain=0.25:0.5] {16*(x-0.5)^2};
    \addplot[ultra thick,blue,name path=A,domain=0.5:1] {24*(x-0.5)^2)};

    \end{axis}
\end{tikzpicture}
\begin{tikzpicture}[scale=0.55]
\pgfplotsset{every tick/.append style={semithick, major tick length=5pt, minor tick length=5pt, black}};
\pgfplotsset{every tick label/.append style={font=\Large}};
    \begin{axis}[
    axis on top=true,
    axis line style = ultra thick,
    axis x line=center,
    xtick={0.5,1}, ytick={0,1},
    axis y line=center,
    ylabel=$v_{M_2}$,
    y label style={at={(-0.35,1.1)}, font=\huge},
    no marks]
    \addplot[smooth,ultra thick,blue,name path=A,domain=0:0.5] {16*(x-0.25)^2};
    \addplot[smooth,ultra thick,blue,name path=A,domain=0.5:1] {16*(x-0.75)^2};

    \end{axis}
\end{tikzpicture}
\end{figure}

 Our analysis does not depend on the exact functional form of indirect utilities in the intervals of strict convexity; it is, however, important that 
 $6\cdot v_{M_1}(0.25)=v_{M_1}(1)$.

To apply the characterization of the sender's optimal equilibrium from Theorem~\ref{th_n_mediators}, we need to find the set of belief distributions ${\mathcal{M}}_1$ in which none of the mediators has an incentive to garble. Since our goal is to apply the characterization of a particular value of $p=0.25$, it is enough to describe the set of those distributions from ${\mathcal{M}}_1$  that have mean $p$. We denote this set by ${\mathcal{M}}_1(p)=\{\mu\in {\mathcal{M}}_1\,:\, \mathbb{E}_\mu[q]=p\}$. 

To find ${\mathcal{M}}_1(p)$, we need to determine ${\mathcal{M}}_2$ first. By definition, ${\mathcal{M}}_2$ is the set of distributions $\mu\in\Delta([0,1])$ such that the second mediator has no incentive to garble. A distribution $\mu$ has this property if and only if any pair of points from its support are affine dominating with respect to $v_{M_2}$ (see Lemma~\ref{lm_FR_as_a_restriction_on_support}). We conclude that
$${\mathcal{M}}_2=\Delta\big([0,\,0.5]\big)\cup \Delta\big([0.5,\,1]\big)\cup\Delta\big(\{0,\,0.5,\,1\}\big).$$
In other words, $\mu$ is supported either on the interval  $[0,\,0.5]$ or on $[0.5,\,1]$ or on the three points $0$, $0.5$, and $1$. Note that ${\mathcal{M}}_2$ is a union of convex sets but not a convex set itself.

Now, we describe ${\mathcal{M}}_1(p)\subset {\mathcal{M}}_2$. By the definition, it consists of distributions $\mu\in {\mathcal{M}}_2$ with mean $p$ such that for any mean-preserving contraction $\nu$ of $\mu$ with $\nu\in {\mathcal{M}}_2$, the first mediator weakly prefers $\mu$ to $\nu$:
\begin{itemize}
\item Consider $\mu$ from the first component $\Delta([0,\,0.5])$ of ${\mathcal{M}}_2$. Any  $\nu\preceq \mu$ is also supported on $[0,\,0.5]$ and hence belongs to ${\mathcal{M}}_2$. Therefore, $\mu$ belongs to ${\mathcal{M}}_1$ if and only if $M_1$ weakly prefers $\mu$ to any mean-preserving contraction. This happens whenever any pair of points from the support of $\mu$ are affine dominating with respect to $v_{M_1}$; equivalently, $\mu$ is either supported on $[0,\,0.25]$ or on $[0.25,\,0.5]$. Such a restriction on the support is compatible with having mean $p=0.25$ only if $\mu$ is a point mass at $p$, i.e., 
$${\mathcal{M}}_1(p)\cap \Delta\big([0,\,0.5]\big)=\{\delta_{p}\}.$$
\item No $\mu$ from the second component $\Delta([0.5,\,1])$ of ${\mathcal{M}}_2$ can have mean $p=0.25$ and hence 
$${\mathcal{M}}_1(p)\cap \Delta\big([0.5,\,1]\big)=\emptyset.$$
\item Finally, let us consider $\mu$ from the third component of ${\mathcal{M}}_2$; i.e., $\mu$ is supported on the
three points $0,$ $0.5$, and $1$.  By definition, $\mu$ belongs to ${\mathcal{M}}_1(p)$ if it has mean $p=0.25$ and for every mean-preserving contraction $\nu$ of $\mu$, where $\nu\in {\mathcal{M}}_2$, the first mediator weakly prefers $\mu$ to $\nu$. Since $\nu\in {\mathcal{M}}_2$, it is supported either on $[0,\,0.5]$, or on $[0.5,\,1]$, or is obtained from $\mu$ by transferring some mass from $0$ and $1$ to $0.5$. The last two cases are excluded: $\nu$ must have mean $p$ as a contraction of $\mu$ and so cannot be supported on $[0.5,\,1]$, while transferring mass to points where $v_{M_1}$ is zero can never be beneficial for $M_1$. It remains to consider $\nu$ supported on $[0,\,0.5]$. The best such $\nu$ from the first mediator's perspective places as much weight on point $0.25$ as possible; i.e., the best $\nu$ is the point mass at $0.25$. Deviation to such $\nu$ is not profitable if and only if 
$v_{M_1}(1)\cdot \mu(\{1\})\geq v_{M_1}(0.25)\cdot 1 $
or, equivalently,
$\mu(\{1\})\geq \frac{1}{6}$.
We conclude that 
$$ {\mathcal{M}}_1(p)\cap \Delta\big(\{0,\,0.5,\,1\}\big)=\Big\{\mu\in \Delta\big(\{0,\,0.5,\,1\}\big)\,:\, \mathbb{E}_\mu[q]=p \mbox{ and } \mu(\{1\})\geq\frac{1}{6}\Big\}.$$
\end{itemize}
Putting all the pieces together we obtain that, for the prior $p=0.25$,
\begin{equation}\label{eq_FR1_example}
{\mathcal{M}}_1(p)=\{\delta_{p}\}\cup \Big\{\mu\in \Delta\big(\{0,\,0.5,\,1\}\big)\,:\, \mathbb{E}_\mu[q]=p \mbox{ and } \mu(\{1\})\geq\frac{1}{6}\Big\}.
\end{equation}
By Theorem~\ref{th_n_mediators}, the sender's optimal payoff is equal to $\cav_{F_1}[v_S](p)=\max\big\{ \mathbb{E}_\mu[v_S(q)],\ \mu\in {\mathcal{M}}_1(p)\big\}$.
The problem of choosing the optimal $\mu\in {\mathcal{M}}_1(p)$ boils down to an elementary finite-dimensional linear program. The distribution $\mu=\delta_{0.25}$ from the first component of ${\mathcal{M}}_1(p)$ gives a payoff of $v_S(0.25)$ and corresponds to a policy revealing no information. Optimizing over $\mu$ from the second component, we conclude that the optimal $\mu$ is supported on the three points 
\begin{equation}\label{eq_mu_opt}
\mu=\frac{2}{3}\delta_0+\frac{1}{6}\delta_{0.5}+\frac{1}{6}\delta_1
\end{equation}
and the corresponding payoff is $\frac{1}{6}v_S(0.5)$.

Thus the optimal sender's value for the prior $p=0.25$ is equal to
$$\max\left\{v_S(0.25),\ \frac{1}{6}v_S(0.5)\right\}.$$
If $v_S(0.25) \geq \frac{1}{6}v_S(0.5)$, the optimal payoff is achieved by revealing no information. For $v_S(0.25) < \frac{1}{6}v_S(0.5)$, the sender's optimal policy induces three different beliefs of the first mediator and hence requires three signals.
\smallskip

We observe several phenomena specific to  persuasion with $n\geq 2$ mediators:
\begin{itemize}
\item  $|\Omega|$ signals are no longer enough for optimal persuasion. Moreover, by restricting to $|\Omega|$ signals, the sender may not be able to guarantee any positive fraction of the optimal payoff. Indeed, the second component of ${\mathcal{M}}_1(p)$ contains no two-point distributions and, hence, the sender restricted to binary signals can only guarantee $v_S(0.25), $ which can be made arbitrarily low while keeping $v_S(0.5)$ unchanged.
\item The sender may benefit from the presence of the second mediator, as we already saw in Example~\ref{ex_interplay}. Indeed, if the second mediator $M_2$ is absent, the best the sender can do for $p=0.25$ is to reveal no information thereby, achieving a payoff of\,\footnote{To apply Theorem~\ref{theorem:main}, describe the set of affine dominating pairs of beliefs with respect to $v_{M_1}$. If  $q_1,q_2$ are affine dominating, then either both points belong to $[0,\, 0.25]$ or both belong to $[0.25,\, 1]$ or one is in $[0,\, 0.25]$ and the other is in $[0.75,\,1]$ (the exact bounds in the last case depend on how convex $v_{M_1}$ is in its intervals of convexity). Any two-point distribution supported on these intervals and having mean $p=0.25$ cannot improve upon the non-revealing policy.} $v_S(0.25)$. Thus, for $v_S(0.25) < \frac{1}{6}v_S(0.5)$, the sender's payoff improves after one adds the second mediator. Intuitively, the presence of the second mediator is beneficial for the following reason. In the absence of $M_2$, the first mediator would garble the sender's optimal distribution $\mu$ given by~\eqref{eq_mu_opt} by moving some mass from points $0$ and $0.5$ to $0.25$. However, in the presence of $M_2$, such garbling is not profitable for $M_1$ as it induces $M_2$ to garble the information even more by shifting the whole mass to $0.5$, i.e., by sending a completely uninformative signal to the receiver, which is the worst outcome for $M_1$. 

\item 
The constraints faced by the sender and  captured by the set ${\mathcal{M}}_1\subset \Delta(\Delta(\Omega))$ cannot be reduced to constraints in $\Delta(\Omega)$ such as constraints on the support of $\mu$, which were enough in the one-mediator case. Indeed, the set ${\mathcal{M}}_1(p)$ from~\eqref{eq_FR1_example} is defined via the constraint $\mu(\{1\})\geq \frac{1}{6}$. Also, we see that ${\mathcal{M}}_1(p)$ may not be convex and may contain several connected components, none of which can be ignored as the optimum may be attained in each of the components, depending on $v_S$.
\end{itemize}
\end{example}

Our analysis raises  several challenging open problems. What is the minimal number of signals sufficient for optimal persuasion with $n\geq 2$ mediators? We do not know even whether a finite number of signals is sufficient for two mediators and a binary state. Are there efficient algorithms approximating the sender's optimal value? Persuasion problems considered in the literature are linear or convex; our problem is non-convex and the feasible set can have several connected components and so standard methods cannot be applied. Can our results be extended from the line graph corresponding to a sequence of mediators to trees or more general networks? As we show in Appendix~\ref{sect_two_direct_mediators}, such an extension will require new insights as the revelation principle underpinning our analysis fails even for the simplest trees.


\bibliography{main}

\appendix

\section{Model, Subgame Perfection, and Other Refinements}\label{app_perfect}

We describe the model paying attention to details, define subgame perfect equilibria together with their $\varepsilon$-relaxation and further refinements needed to handle discontinuous utilities, and discuss the role of technical assumptions. 

There are $n+2$ agents: a sender $S$, a sequence of $n\geq 1$ mediators $M_1,\ldots, M_n$, and a receiver~$R$. It will also be convenient to refer to them as agents $0,1,\ldots, n,n+1$. The set of states $\Omega$ is finite and is endowed with a prior $p\in\Delta(\Omega)$.  The receiver's set of actions $A$ is a measurable space. Agents' utilities $u_i$, $i=0,\ldots, n+1,$ are bounded measurable functions on $\Omega\times A$.

 A signaling policy of the sender is a pair $F_0=(S_0, f_0)$ consisting of a set of signals $S_0$ and a map $f_0:\, \Omega\to\Delta(S_0)$ that defines a distribution of the sender's signal $s_0\in S_0$
for each possible realization of the state $\omega\in \Omega$. A mediator $i=1,\ldots, n$ observes the signal $s_{i-1}$ sent by her predecessor and sends a signal $s_i$ to her successor. Hence, mediator $i$'s signaling policy is a pair  $F_i=(S_i, f_i)$, where $f_i\colon S_{i-1}\to\Delta(S_i)$. 
A policy of the receiver maps the last mediator's signal to a randomized action $a\in A$, i.e., $F_{n+1}=(A,f_{n+1})$, where  $f_{n+1}\colon S_{n}\to\Delta(A)$. The sets of signals $S_i$ are assumed to be arbitrary measurable spaces and functions $f_i$ are such that $f_i(s)(B)$ is measurable in $s$ for any measurable set $B$ (such $f_i$ are called Markov kernels). Both $S_i$ and $f_i$ are components of agent $i$'s strategic choice.\footnote{Fixing a rich enough set of signals $S_i$---say, $[0,1]$ or, more generally, any uncountable standard Borel space---leads to an equivalent model.}

Agents select their policies sequentially and so agent $i$'s policy choice can be affected by choices made by agents $0,1,\ldots, i-1$. The history $h_i$ observed by agent $i$ consists of all the policies chosen by predecessors: $h_i=(F_j)_{j=0}^{i-1}$ (for convenience, $h_0=\{\emptyset\}$). A strategy $\sigma_i$ of agent $i$ specifies a policy $F_i$ for each history $h_i,$ i.e., $F_i=\sigma_i(h_i)$.
A profile of all agents' strategies $(\sigma_0,\ldots, \sigma_{n+1})$ determines inductively the profile of policies $F_{i}=\sigma_i\big((F_j)_{j=0}^{i-1}\big)$ referred to as the equilibrium path. 

A profile of policies $F_0,\ldots, F_{n+1}$ and the prior $p$ induce the joint distribution of the state $\omega\in \Omega$, signals $(s_0,\ldots, s_n)\in S_0\times \ldots \times S_n$, and the action $a\in A$. The expectation with respect to this distribution is denoted by $\mathbb{E}=\mathbb{E}_{(F_0,\ldots, F_{n+1})}$. Agent $i$'s expected payoff is given by $\mathbb{E}_{(F_0,\ldots, F_{n+1})}[u_k(\omega,a)]$.

\begin{definition}
Given $\varepsilon\geq0$, some policies $G_0,\ldots, G_{i-1}$ of agent~$i$'s predecessors, and strategies $\sigma_{i+1},\ldots, \sigma_{n+1}$ of $i$'s successors, a policy $F_i$ of agent $i$ is called an $\varepsilon$-best reply to the history $(G_j)_{j=0}^{i-1}$ if no other $F_i'$ can increase $i$'s payoff by more than $\varepsilon$ under the assumption that $i$'s successors follow the policy choice prescribed by their strategies.\footnote{Formally, let $F_{i+1},\ldots, F_{n+1}$ be the policies chosen by agents $i+1,\ldots, n+1$ if the first agents select $G_0,\ldots, G_{i-1}$ and $F_i$ and let $F_{i+1}',\ldots, F_{n+1}'$ be the corresponding policies if the game starts with $G_0,\ldots, G_{i-1}$ and $F_i'$. Then $F_i$ is an $\varepsilon$-best reply at the history $(G_j)_{j=0}^{i-1}$ if for any $F_i'$
  \begin{equation}\label{eq_epsilon_best_reply}
      \mathbb{E}_{(G_0,\ldots, G_{i-1},F_{i}',F_{i+1}',\ldots,F_{n+1}')}\big[u_i(\omega,a)\big]\leq
      \mathbb{E}_{(G_0,\ldots, G_{i-1},F_{i},F_{i+1},\ldots,F_{n+1})}\big[u_i(\omega,a)\big]+\varepsilon.  \end{equation}
      } If $\varepsilon=0$, the policy  $F_i$ is called a best reply.
\end{definition}
As usual in Bayesian persuasion, we assume that agents cannot use non-credible threats to incentivize the desired behavior of the predecessors. This assumption is captured by the concept of a subgame perfect equilibrium  requiring that agents' strategies be best replies to all histories (including  those that never arise on the equilibrium path). 

A subgame perfect equilibrium is  defined formally below. We add an $\varepsilon$-slack in it to deal with discontinuous utilities since for such utilities best replies may fail to exist.  The definition is standard except for two refinements restricting the ways agents can use tie-breaking to punish predecessors. To introduce the refinements, we need the following notation.  

A policy $F_i=(S_i,f_i)$ of a mediator $i=1,\ldots, n$ is the full-revelation policy if the mediator transmits unchanged the signal $s_{i-1}\in S_{i-1}$ received from her predecessor, i.e., $S_i=S_{i-1}$ and $f_i(s_i)=\delta_{s_i}$, where $\delta_x$ denotes the point mass at $x$. We will require that all the mediators use full-revelation policies unless they can strictly improve upon such policies  by more than $\varepsilon$.

Given a joint distribution of the state $\omega$ and signals $s_0,\ldots, s_n$, we denote  agent $i$'s belief induced by the observed signal by $p_i$, $i=1,\ldots, n+1$, i.e., $p_{i,k}=\mathbb{P}(\omega=k\mid s_{i-1})$ for $k\in \Omega$. In other words, $p_i$ is the  distribution of $\omega$ conditional on $s_{i-1}$. The belief $p_i$ is itself a random variable with values in $\Delta(\Omega)$ as it depends on $s_{i-1}$. We note that to compute $p_i$ given $s_{i-1}$, the agent $i$ only needs to know the policies chosen by predecessors $F_0,\ldots, F_{i-1}$. We will require that the receiver's action depend on the policies of other agents inasmuch they determine her belief $p_{n+1}$.
\begin{definition}\label{def_subgame}
  For $\varepsilon\geq 0$, a profile of strategies $(\sigma_0,\ldots, \sigma_{n+1})$ is a \ed{subgame perfect $\varepsilon$-equilibrium with refined tie-breaking ($\varepsilon$-RTSPE, henceforth)} if the following requirements are satisfied:
  \begin{itemize}
  \item \emph{Subgame perfection:} For any agent $i=0,\ldots, n+1$ and any history $h_i=(G_j)_{j=0}^{i-1}$, the policy $F_i=\sigma_i(h_i)$ is an $\varepsilon$-best reply to $h_i$.
  \item \emph{Full-revelation refinement:} For any mediator $i=1,\ldots, n$ and any history $h_i$, if the full-revelation policy is an $\varepsilon$-best reply to $h_i$, then $\sigma_i(h_i)$ is the full-revelation policy.
  \item \emph{Belief-driven receiver refinement:}
  The receiver's action is a function of her belief $p_{n+1}$ induced by the last mediator's signal $s_n$ for any history. Formally, for any $h_{n+1}=(G_j)_{j=0}^{n}$, the receiver's policy $F_{n+1}=(A,f_{n+1})=\sigma_{n+1}(h_{n+1})$ can be factorized as follows: $f_{n+1}(s_n)=\widehat{f}_{n+1}(p_{n+1}(s_n))$, where a function $\widehat{f}_{n+1}\colon \Delta(\Omega)\to\Delta(A)$ does not depend on history $h_{n+1}$.
\end{itemize}
We will refer to the case of $\varepsilon=0$ as \ed{a subgame perfect equilibrium with refined tie-breaking (RTSPE).}
\end{definition}

\subsection{The role of refinements} Let us discuss the intuition for \ed{full-revelation and belief-driven}   refinements. Without them, subgame perfection allows an agent indifferent between several choices to pose threats  by making the tie-breaking dependent on the choices made by predecessors. The following example demonstrates that large sets of indifference may lead to unnatural equilibria if we drop the full-revelation refinement.\footnote{An alternative workaround would be to assume that the game is generic and best replies are unique. However, this assumption is too restrictive for games with large sets of strategies.}
\begin{example}[Punishment through indifference]
Consider a problem with two mediators $M_1$ and $M_2$, where the second mediator's utility is constant in both $\omega$ and $a$.  Any strategy of $M_2$ is compatible with subgame perfection. In particular, $M_2$ may decide not to transmit any information to the receiver (i.e., to send a dummy signal independent of~$\omega$) unless the policies chosen by the sender and $M_1$ fully reveal the state. 

If keeping the receiver completely uninformed is the worst outcome for the sender and the first mediator, such a strategy of $M_2$ induces full revelation even if full revelation does not occur in a problem where $M_2$ is absent. We see that without the full-revelation refinement, the presence of an indifferent mediator may alter the equilibrium even though it would be natural to require the equilibrium not to be sensitive to the presence of completely indifferent dummy agents. The full-revelation refinement enforces the latter natural behavior.
\end{example}

The refinement of the receiver's behavior is needed for a similar reason. To see that this refinement is intuitive, consider the receiver's expected payoff conditional on $s_n$ as a function of the receiver's action $a$. The payoff can be represented as follows
\begin{equation}\label{eq_receivers_payoff}
\mathbb{E}[u_{n+1}(\omega,a)\mid s_n]= \mathbb{E}\Big[\sum_{k\in \Omega} p_{n+1,k}(s_n)
\cdot u_{n+1}(k,a)\mid s_n\Big].
\end{equation}
Given a belief $q\in \Delta(\Omega)$, denote by $\Lambda_\varepsilon(q)\subset \Delta(A)$ the set of distributions $\lambda$ over the set of actions~$A$ such that
$$\int_A \sum_{k\in \Omega} q_k\cdot  u_{n+1}(k,a)\, \dd\lambda(a)\geq \sup_{a\in A} \int_A \sum_{k\in \Omega} q_k u_{n+1}(k,a)\, \dd\lambda(a)-\varepsilon.$$
From~\eqref{eq_receivers_payoff}, we conclude that the
receiver's policy $F_{n+1}=(A,f_{n+1})$ is an $\varepsilon$-best reply to a history $h_{n+1}$ if and only if the distribution $f_{n+1}(s_n)$ of the receiver's actions belongs to $\Lambda_\varepsilon(p_{n+1}(s_n))$. We see that the set of action distributions constituting an $\varepsilon$-best reply is determined by the posterior $p_{n+1}$. For example, if $\varepsilon=0$ and a best reply exists and is unique, i.e., $\Lambda_0(q)$ is a singleton for any\footnote{For example, $\Lambda_0(q)$ is a singleton if $A$ is a compact convex subset of $\mathbb{R}^d$ and $u_{n+1}$ is a strictly convex continuous function of $a\in A$.} $q$, then any receiver's behavior satisfying subgame perfection is automatically belief-driven. In general, the requirement of belief-driven behavior means that the choice of the action distribution from $\Lambda_\varepsilon(p_{n+1})$ cannot depend on the information that is payoff-irrelevant to the receiver.

\subsection{Indirect utilities} 
The assumption that the receiver's behavior is belief-driven makes it possible to define the indirect utilities of all the agents, i.e., to treat their utilities as functions of the  receiver's induced belief.

Recall that a strategy of a belief-driven receiver is captured by a function $\widehat{f}_{n+1}\colon \Delta(\Omega)\to \Delta (A)$ specifying the action distribution for each belief. Given $\widehat{f}_{n+1}$, the indirect utility of agent $i$ for a belief $q\in \Delta(\Omega)$ is defined by $$v_{i}(q)=\int_A\sum_{k\in \Omega} q_k\cdot  u_i(k,a)\dd \lambda(a),$$ where $\lambda=\widehat{f}_{n+1}(q)$. By the definition, the expected utility of agent~$i$ can be represented through her indirect utility as $\mathbb{E}[v_{i}(p_{n+1})]$, where $p_{n+1}$ is the receiver's belief.

In our analysis, we treat the indirect utilities of the sender and the mediators $v_0,v_1,\ldots, v_n$ as  primitives of the model. In other words, we assume that both $\varepsilon\geq 0$ and  a belief-driven strategy $\widehat{f}_{n+1}$ of the receiver are fixed and $\widehat{f}_{n+1}$ is an $\varepsilon$-best reply to any history. We note that such a strategy $\widehat{f}_{n+1}$ exists for any $\varepsilon>0$ as the set $\Lambda_\varepsilon(q)$ is non-empty for any $q\in \Delta(\Omega)$. For $\varepsilon=0$, a best reply may fail to exist unless some additional assumptions are imposed (e.g., the compactness of $A$ and the upper semicontinuity of $u_{n+1}$ in $a$). Hence, for $\varepsilon=0$, consideration of indirect utilities is not without loss of generality and relies on an implicit assumption that the receiver's best reply exists.

\subsection{The sender's problem}

The problem is given by a prior $p\in \Delta(\Omega)$,  a parameter $\varepsilon\geq 0$, and indirect utilities $v_0$ of the sender  and $v_{1},\ldots, v_{n}$ of mediators.

The sender's goal is to maximize her expected payoff $$\mathbb{E}[v_0(p_{n+1})]$$
over all \ed{$\varepsilon$-RTSPE}. The optimal value is denoted by $V_S^\varepsilon(p)$, i.e.,
$$V_S^\varepsilon(p)= \sup_{\footnotesize{\mbox{$\varepsilon$-RTSPE}}} \mathbb{E}[v_0(p_{n+1})].$$
If supremum is taken over an empty set, we put $V_S^\varepsilon(p)=-\infty$.

As we will see, the set of \ed{$\varepsilon$-RTSPE} is non-empty for any $\varepsilon>0$. For $\varepsilon=0$, equilibria exist under continuity assumptions on indirect utilities.


\section{Blackwell's Theory}\label{subsec_Blackwell}
We refer to the sender $S$, the mediators $M_1,\ldots, M_n$, and the receiver $R$ as agents $0,\ldots,n+1$. Recall that a prior $p\in \Delta(\Omega)$ and policies  $F_0,\ldots, F_n$ of the first $n$ agents defines a joint distribution of $\omega$ and signals $s_0,\ldots, s_n$. The belief of agent $i$ about $\omega$ induced by the observed signal $s_{i-1}$ is denoted by $p_i=p_{i}(s_{i-1})\in \Delta(\Omega)$, where $p_{i,k}=\mathbb{P}(\omega=k\mid s_{i-1})$.

Let $\mu_i\in \Delta(\Delta(\Omega))$ be the distribution of agent $i$'s belief. The expected payoff of agent~$i$ can be represented through her indirect utility and the receiver's belief distribution as $\mathbb{E}_{\mu_{n+1}}[v_i(q)]=\int_{\Delta(\Omega)} v_i(q)\dd \mu_{n+1}(q)$.
As we explain below, one can abstract from the details of policies used by agents and keep track of the induced  sequence of belief distributions $\mu_1,\ldots, \mu_{n+1}$ only.

A distribution $\mu\in \Delta(\Delta(\Omega))$ is
a mean-preserving spread of $\nu\in \Delta(\Delta(\Omega))$, denoted by $\mu\succeq \nu$,
if\,\footnote{An equivalent definition is that $\mu\succeq\nu$ if for any convex function $\varphi$ on $\Delta(\Omega)$ we have $\int_{\Delta(\Omega)} \varphi(q)\dd\mu(q)\geq \int_{\Delta(\Omega)}\varphi(q)\dd\nu(q)$.} there exists a pair of random variables $q_1$ distributed according to $\nu$ and $q_2$ distributed according to $\mu$ defined on the same probability space and forming a martingale, i.e., such that $\mathbb{E}[q_2\mid q_1]=q_1$. Analogously, we call $\nu$ a mean-preserving contraction of $\mu$. 

Blackwell~\cite{blackwell1950comparison} proved that for a signal $s_{\mathrm{in}}\in S_{\mathrm{in}}$ inducing some distribution of beliefs $\mu$, there exists a signaling policy $F=(S_{\mathrm{out}},f\colon S_{\mathrm{in}}\to \Delta(S_{\mathrm{out}}))$ such that a signal $s_{\mathrm{out}}$ induces a distribution $\nu$ if and only if $\mu \succeq \nu$. Informally, a mean-preserving contraction of a belief distribution corresponds to garbling the information.

Applying Blackwell's result to our model, we 
conclude that a sequence of belief distributions $\mu_1,\ldots, \mu_{n+1}$ corresponds to some profile of signaling policies $F_0,\ldots, F_n$ if and only if $\mu_0\succeq\mu_1\succeq \ldots\succeq \mu_{n+1}$.
Here $\mu_0$ denotes the distribution of the sender's belief about $\omega$ induced by observing the realization of $\omega$, i.e., $\mu_0=\sum_{k\in\Omega} p_k\delta_k$, where $\delta_k$ is the point mass at state $k$.

The set of $\mu_1$ such that $\mu_0=\sum_{k\in\Omega} p_k\delta_k\succeq \mu_1$ has a simple structure. It is determined by the martingale property: $\mu_0\succeq \mu_1$ if and only if $\mathbb{E}_{\mu_1}[q]=p$, i.e., the mean of $\mu_1$ is equal to the prior (by the splitting lemma  of Aumann and Maschler~\cite{aumann1995repeated}). We conclude that a necessary condition for $\mu_1,\ldots, \mu_{n+1}$ to correspond to some profile of signaling policies is that $\mathbb{E}_{\mu_i}[q]=p$ for all $i=1,\ldots,n$.

\section{Proof of Theorem \ref{theorem:main}}\label{app_one_receiver}
Here we formulate and prove a general version of Theorem \ref{theorem:main} allowing for irregular indirect utilities. 

There is one mediator $M$, and the indirect utilities of the sender and the mediator are $v_S$ and $v_M$, respectively. We do not impose any assumptions on $v_S$ and $v_M$ except for measurability and boundedness. For irregular $v_M$, the mediator may not have a best reply to some policies of the sender. To deal with such irregularities, we use \ed{subgame perfect $\varepsilon$-equilibrium with refined tie-breaking ($\varepsilon$-RTSPE)} introduced in Appendix~\ref{app_perfect}. Recall that, for fixed $\varepsilon\geq 0$,  we defined $V_S^\varepsilon(p)$ as the supremum of the sender's payoff over all \ed{$\varepsilon$-RTSPE} (if the set of such equilibria is empty,  $V_S^\varepsilon(p)=-\infty$). 

Recall the definition~\eqref{eq:max} of constrained concavification for an upper semicontinuous $v_S$ and a continuous  $v_M$: 
\begin{align}
\cav_{{D}}\big[v_S\big](p)=
\notag\max\Big\{ \sum_{k=1}^{{|\Omega|}} \alpha_k v_S(q_k) \ \Big| \ 
(q_1,\ldots,q_{{|\Omega|}})\in {D}, \ \alpha\in \Delta(\Omega), \ \sum_{k=1}^{{|\Omega|}}\alpha_k q_k=p \Big\},
\end{align}
where ${D}$ consists of collections $(q_1,\ldots,q_{{|\Omega|}})\in \Delta(\Omega)^\Omega$ that are affine dominating with respect to the mediator's indirect utility, i.e.,
$\sum_{k=1}^{{|\Omega|}} \alpha_k\cdot v_M\big(q_k\big)\geq v_M\left(\sum_{k=1}^{{|\Omega|}} \alpha_k\cdot q_k\right)$ for all non-negative  $\alpha_1,\ldots, \alpha_{|\Omega|}$ that sum up to one.
To handle discontinuities, we generalize the definition as follows. Let $\mathcal{M}$ be a subset of $\Delta(\Delta(\Omega))$. The constrained concavification of $v_S$ with respect to  $\mathcal{M}$ is given by
\begin{equation}\label{eq_constrained_concavification_extended}
\cav_{\mathcal{M}}\big[v_S\big](p)=
\sup\left\{ \mathbb{E}_\mu [v_S(q)] \  \Big| \ \mu\in\mathcal{M}, \  \mathbb{E}_\mu[q]=p \right\},
\end{equation}
where $\mathbb{E}_\mu[f(q)]$ denotes the expected value of $f(q)$ with $q$ distributed according to $\mu$.
Note that $\cav_{\mathcal{M}}\big[v_S\big]=\cav_{{D}}\big[v_S\big]$ if we define $\mathcal{M}$ as the set of all distributions of the form $\sum_{k=1}^{{|\Omega|}} \alpha_k \delta_{q_k}$ where $(q_1,\ldots,q_{{|\Omega|}})$ belong to ${D}$.

Using the notion of a mean-preserving contraction (see Appendix~\ref{subsec_Blackwell}), consider the set of belief distributions 
such that no garbling can increase the mediator's utility by more than $\varepsilon$:
\begin{equation}\label{eq_FR_epsilon_one_mediator}
{\mathcal{M}}^\varepsilon=\Big\{\mu\in \Delta(\Delta(\Omega)) \ \Big|\  \Big(\nu\preceq \mu\Big)\Longrightarrow  \mathbb{E}_\mu\big[v_M(q) \big]\geq \mathbb{E}_{\nu}\big[v_M(q) \big]-\varepsilon\Big\}.
\end{equation}
This set contains the point mass concentrated at $p$ for any $p\in \Delta(\Omega)$ and, hence, $\cav_{{\mathcal{M}}^\varepsilon}[v_S](p)$ is well defined for any $p$ since the maximization is over a non-empty set. 
\begin{theorem}[Generalized version of Theorem \ref{theorem:main}]\label{theorem:appendix_one_mediatior}
Assuming that indirect utilities $v_S$ and $v_M$ are bounded measurable functions, the sender's optimal payoff in a \ed{subgame perfect $\varepsilon$-equilibrium with refined tie-breaking} satisfies
$$V_S^\varepsilon(p)=\cav_{{\mathcal{M}}^\varepsilon}\big[v_S\big](p)$$ for any $\varepsilon>0$.

If, additionally, $v_S$ is  upper semicontinuous and $v_M$ is continuous, the optimal equilibrium exists and $$V_S^0(p)=\cav_{{\mathcal{M}}^0}\big[v_S\big](p)=\cav_{{D}}\big[v_S\big](p).$$
\end{theorem}

The proof of Theorem~\ref{theorem:appendix_one_mediatior} is split into several lemmas. \ed{Lemma~\ref{lm_upper_varepsilon} shows that the sender's optimal payoff is upper-bounded by $\cav_{{\mathcal{M}}^\varepsilon}\big[v_S\big]$, Lemma~\ref{lm_lower_varepsilon} gives the matching  lower bound. Essentially, these two lemmas demonstrate that it is enough to look at those equilibria where the only garbling is done by the sender and the mediator has no $\varepsilon$-profitable deviations from  revealing the information fully to the receiver. The set of belief distributions that the sender can induce in this class of equilibria is exactly the set ${\mathcal{M}}^\varepsilon$. Lemma~1 from online Appendix of (Lipnowski~et al.~\cite{lipnowski2020attention}) implies that for upper semicontinuous $v_S$, continuous $v_M$, and $\varepsilon=0$, it is enough
to maximize over those distributions from $\mathcal{M}^0$ that are supported on at most $|\Omega|$ points. For reader's convenience, we give a self-contained proof of this fact in Lemmas~\ref{lm_FR_as_a_restriction_on_support} and~\ref{lm_two_concavifications_are_the_same}. As a result, we obtain that
$\cav_{{\mathcal{M}}^0}\big[v_S\big]$ coincides with $\cav_{{D}}\big[v_S\big]$.}

\begin{lemma}\label{lm_upper_varepsilon}
Let $F_S=(S_S,\,  f_S\colon\Omega\to \Delta(S_S))$ be a policy chosen by the sender and $F_M=(S_M,\, f_M\colon S_S\to \Delta(S_M))$ be a mediator's $\varepsilon$-best reply$\,$\footnote{Recall that $F_M$ is an $\varepsilon$-best reply to $F_S$ if any other policy $F_M'=(S_M',\, f_M'\colon\, S_S\to \Delta(S_M'))$ cannot increase the mediator's payoff by more than $\varepsilon$ provided that the sender keeps her policy $F_S$ unchanged.} for some $\varepsilon\geq 0$. Then the sender's expected payoff cannot exceed $\cav_{{\mathcal{M}}^\varepsilon}\big[v_S\big](p).$
\end{lemma}
\begin{proof}[Proof of Lemma~\ref{lm_upper_varepsilon}]
The pair of policies $F_S$ and $F_M$ and the prior $p$ induce a joint distribution of $\omega$, $s_S$, and $s_M$. Let $p_R\in \Delta(\Omega)$ be the receiver's  belief after observing the  signal $s_M$ and $\mu\in \Delta(\Delta(\Omega))$ be the distribution of $p_R$. The sender's payoff is equal to $\mathbb{E}[v_S(p_R)]$ or, equivalently, $\mathbb{E}_\mu[v_S(q)]$ with $q$ distributed according to $\mu$. Hence, to prove the lemma, it is enough to show that  $\mu$ satisfies $\mathbb{E}_\mu[q]=p$ and belongs to ${\mathcal{M}}^\varepsilon$. The first requirement is satisfied by the martingale property of posterior beliefs: the expected posterior is equal to the prior and so $\mu$ has the right mean, $\mathbb{E}_\mu[q]=p$. It remains to show that for any $\nu\preceq \mu$, the mediator's payoff for $\nu$ cannot exceed that for $\mu$ by more than $\varepsilon$, i.e., $\mathbb{E}_\mu\big[v_M(q) \big]\geq \mathbb{E}_{\nu}\big[v_M(q) \big]-\varepsilon$. By Blackwell's theorem, for any $\nu\preceq \mu$ there is a policy $F_M'=(S_M',\, f_M'\colon S_S\to \Delta(S_M'))$ whereby the mediator induces the distribution $\nu$ of the receiver's beliefs. Since $F_M$ is an $\varepsilon$-best reply, no policy $F_M'$ can increase the mediator's utility by more than $\varepsilon$. Therefore, $\mu$ belongs to ${\mathcal{M}}^\varepsilon$ and thus the sender's payoff does not exceed~$\cav_{{\mathcal{M}}^\varepsilon}\big[v_S\big](p)$. 
\end{proof}
From Lemma~\ref{lm_upper_varepsilon}, we conclude that the sender's optimal payoff $V_S^\varepsilon(p)$ cannot exceed $\cav_{{\mathcal{M}}^\varepsilon}\big[v_S\big]$.
The next lemma provides a lower bound.
\begin{lemma}\label{lm_lower_varepsilon}
For any $\varepsilon>0$, the sender's optimal payoff in \ed{an $\varepsilon$-RTSPE} satisfies 
\begin{equation}\label{eq_lower_bound}
V_S^\varepsilon(p)\geq \cav_{{\mathcal{M}}^\varepsilon}\big[v_S\big](p).
\end{equation}

If $v_S$ is upper semicontinuous and bounded and $v_M$ is continuous, the bound~\eqref{eq_lower_bound}  also holds for $\varepsilon=0$ and there exists \ed{an $\varepsilon$-RTSPE} such that the sender's payoff is at least $\cav_{{\mathcal{M}}^0}\big[v_S\big](p)$.
\end{lemma}
\begin{proof}[Proof of Lemma~\ref{lm_lower_varepsilon}]
To prove~\eqref{eq_lower_bound}, it is enough to show that for any $\delta>0$, there is \ed{an $\varepsilon$-RTSPE} where the sender's payoff is at least $\cav_{{\mathcal{M}}^\varepsilon}\big[v_S\big](p)-\delta$. Pick a distribution $\mu^\delta\in {\mathcal{M}}^\varepsilon$ such that $\mathbb{E}_{\mu^\delta}[q]=p$ and $\mathbb{E}_{\mu^\delta} [v_S(q)]$ is at least  $\cav_{{\mathcal{M}}^\varepsilon}\big[v_S\big](p)-\min\{\varepsilon,\delta\}$. 

We construct the desired equilibrium as follows. The sender's strategy is such that her signal $s_S$ induces the distribution of beliefs $\mu^\delta$ of the mediator. The mediator is playing an $\varepsilon$-best reply to this policy or any other policy that the sender may have chosen. 
An $\varepsilon$-best reply can be constructed explicitly.
If the distribution of the mediator's beliefs $\mu_M$ induced by the sender's signal $s_S$ belongs to ${\mathcal{M}}^\varepsilon$, the mediator transmits the signal unchanged; i.e., $f_M: \, S_S\to \Delta(S_S)$ maps the signal $s_S$ to the point mass at $s_S$. 
If $\mu_M$ is outside of ${\mathcal{M}}^\varepsilon$, the mediator can garble the sender's signal and improve her own payoff by more than $\varepsilon$; i.e., there is $\nu\preceq \mu_M$ such that $\mathbb{E}_{\mu_M} [v_M(q)]<\mathbb{E}_{\nu} [v_M(q)]-\varepsilon$. The mediator selects a policy $f_M: \, S_S\to \Delta(S_M)$ in order to induce a distribution $\nu=\nu^\varepsilon(\mu_M)$ that gives her an $\varepsilon$-optimal payoff; i.e., $\mathbb{E}_{\nu^\varepsilon(\mu_M)} [v_M(q)]+\varepsilon\geq \sup\left\{\mathbb{E}_{\nu} [v_M(q)], \ \ \nu\preceq \mu_M\right \}$.

Now we check that the described strategies form \ed{an $\varepsilon$-RTSPE.} By the definition of ${\mathcal{M}}^\varepsilon$ and the choice of $\nu^\varepsilon(\mu_M)$, no deviation of the mediator can improve her payoff by more than $\varepsilon$. Hence, the mediator's  strategy is an $\varepsilon$-best reply to any strategy of the sender. By Lemma~\ref{lm_upper_varepsilon}, the sender's payoff is bounded from above by $\cav_{{\mathcal{M}}^\varepsilon}\big[v_S\big](p)$.
The mediator does not garble the sender's signal unless garbling is strictly profitable; i.e., the full revelation refinement  holds (see Definition~\ref{def_subgame}). In particular, the mediator does not garble the signal inducing the distribution $\mu^\delta$ as $\mu^\delta$ belongs to ${\mathcal{M}}^\varepsilon$. Thus, by the choice of $\mu_\delta$, the sender's payoff is at least $\cav_{{\mathcal{M}}^\varepsilon}\big[v_S\big](p)-\min\{\varepsilon,\delta\}$. We conclude that the sender has no deviations improving her payoff by more than $\varepsilon$.

To summarize, we constructed \ed{an $\varepsilon$-RTSPE} with the sender's payoff at least $\cav_{{\mathcal{M}}^\varepsilon}\big[v_S\big](p)-\delta$. As $\delta$ was arbitrary, we obtain the desired bound~\eqref{eq_lower_bound}.

Now consider the case of bounded upper semicontinuous $v_S$ and  continuous $v_M$. The continuity of $v_M$ ensures that the set ${\mathcal{M}}^\varepsilon$  as well as the set of distributions $\nu$ are compact in the weak topology and, in particular, the best-reply distribution $\nu^\varepsilon(\mu_M)$ with $\varepsilon=0$ exists  for any $\mu_M$. The upper semicontinuity of $v_S$ implies the upper semicontinuity of 
$\mathbb{E}_\mu [v_S(q)]$ as a function of $\mu$. An upper semicontinuous functional attains its maximum on a compact set and hence the sender's optimal distribution $\mu^\delta$ exists for $\delta=0$. Therefore, we can plug $\varepsilon=\delta=0$ into the above construction and  ensure that all the optima are attained. Thus, for continuous $v_M$ and  upper semicontinuous $v_S$,  we obtain \ed{an RTSPE} with a sender's payoff of at least $\cav_{{\mathcal{M}}^0}\big[v_S\big](p)$. 
\end{proof}
Lemmas~\ref{lm_upper_varepsilon} and~\ref{lm_lower_varepsilon} give matching upper and lower bounds on the sender's best payoff $V_S^\varepsilon(p)$ and imply the first part of Theorem~\ref{theorem:appendix_one_mediatior}.  It remains to show that  for regular indirect utilities, the concavification $\cav_{{\mathcal{M}}^0}\big[v_S\big]$ can be computed as a maximization  over distributions supported on $|\Omega|$ affine dominating points, i.e.,
$\cav_{{\mathcal{M}}^0}\big[v_S\big]=\cav_{{D}}\big[v_S\big]$. This is done in two steps. First, in Lemma~\ref{lm_FR_as_a_restriction_on_support}, we show that whether or not a distribution $\mu$ belongs to ${\mathcal{M}}^0$ is determined by the support of $\mu$. Next, Lemma~\ref{lm_two_concavifications_are_the_same}  leverages this observation to show the desired equality via an extreme-point argument.

Let us extend the notion of affine domination from collections  $q_1,\ldots,q_{|\Omega|}$  of posteriors to arbitrary closed subsets of $\Delta(\Omega)$.
\begin{definition}\label{def:see2}
 A closed subset $\mathcal{D}\subset \Delta(\Omega)$ is \emph{affine dominating with respect to a continuous function $f:\, \Delta(\Omega)\to \mathbb{R}$} if for every measure $\mu\in\Delta(\Delta(\Omega))$ such that $\mu(\mathcal{D})=1$, it holds that $\mathbb{E}_\mu[f(q)]\geq f\big(\mathbb{E}_\mu[q]\big)$.
\end{definition}
Recall that the support $\supp[\mu]$ of a distribution $\mu$ is the minimal closed set of full measure.

\begin{lemma}\label{lm_FR_as_a_restriction_on_support}
For continuous $v_M$, a distribution $\mu\in\Delta(\Delta(\Omega))$ belongs to ${\mathcal{M}}^0$ if and only if the support $\supp[\mu]$ is affine  dominating with respect to $v_M$.
\end{lemma}
\begin{proof}  \emph{The ``if'' direction.}
We assume that $\supp[\mu]$ is affine dominating and show that, for any $\nu\preceq \mu$, the mediator prefers $\mu$ to $\nu$, i.e.,  $\mathbb{E}_\mu\big[v_M(q) \big]\geq \mathbb{E}_{\nu}\big[v_M(q)\big]$.
 By the definition of a mean-preserving spread, there exists a martingale $X_1,X_2$ (on a natural filtration) such that $X_1$ is distributed according to $\nu$, and $X_2$ is distributed according to $\mu$. Hence, the mediator's payoff for $\nu$ can be represented as follows:
$$\mathbb{E}_{\nu}\big[v_M(q)\big]=\mathbb{E}\big[v_M(X_1)\big].$$
On the other hand,
$$\mathbb{E}_{\mu}\big[v_M(q)\big]=\mathbb{E}\big[v_M(X_2)\big]\geq \mathbb{E}\Big[v_M\big(\mathbb{E}[X_2\mid X_1]\big)\Big]=\mathbb{E}\big[v_M(X_1)\big],$$
where the inequality holds since  $X_2$ is supported on the affine dominating set and the last equality follows from the fact that $(X_1,X_2)$ is a martingale.
We conclude that $\mathbb{E}_\mu\big[v_M(q) \big]\geq \mathbb{E}_{\nu}\big[v_M(q)\big]$ and thus any distribution supported on an affine dominating set belongs to ${\mathcal{M}}^0$.
\smallskip 

\noindent \emph{The ``only if'' direction.} 
Assume that  $\supp[\mu]$ is not affine dominating and show that there is $\nu\preceq \mu $ preferred by the mediator to  $\mu$, i.e.,  $\mathbb{E}_\mu\big[v_M(q) \big]< \mathbb{E}_{\nu}\big[v_M(q)\big]$.

Since the condition of affine domination is violated, we can find a distribution $\tau$ such that $\supp[\tau]\subset \supp[\mu]$ and $\mathbb{E}_\tau\big[v_M(q)\big] < v_M\big(\mathbb{E}_\tau[q]\big)$.
Let us show that one can find such a distribution with the additional property that it has a bounded density with respect to~$\mu$.

By the continuity of $v_M$, we can find $\tau$ with a finite support (indeed, start from a general $\tau$ and approximate it by a distribution supported on an $\varepsilon$-net of $\Delta(\Omega)$ for small enough $\varepsilon$). Let $q_1,\ldots,q_m$ be the posteriors from the support of $\tau$ and $\alpha_1,\ldots,\alpha_m$ be the respective weights, i.e., $\tau=\sum_{i=1}^m \alpha_i\cdot \delta_{q_i}$. Once again, by the continuity of $v_M$, we can find disjoint open neighborhoods $U_{i}$ of $q_i$, $i=1,\ldots, m$, such that for any $q_i'\in U_i$ the distribution $\tau'=\sum_{i=1}^m \alpha_i\cdot \delta_{q'_i}$ also has the property $\mathbb{E}_{\tau'}\big[v_M(\mu)\big] < v_M\big(\mathbb{E}_{\tau'}[q]\big)$  (points may be different but the weights remain the same). Define a distribution $\overline{\tau}$ as follows:
$$\overline{\tau}=\sum_{i=1}^m \frac{\alpha_i}{\mu(U_i)}\mu\vert_{U_i},$$
where $\mu\vert_{U_i}$ denotes the restriction of $\mu$ to $U_i$, i.e., $\mu\vert_{U_i}(B)=\mu(B\cap U_i)$ for every Borel set $B$. Note that the denominators $\mu(U_i)\ne 0$ since $q_i$ belong the support of $\mu$.

By the construction of $\overline{\tau}$, the inequality
$\mathbb{E}_{\overline{\tau}}\big[v_M(q)\big] < v_M\big(\mathbb{E}_{\overline{\tau}}(q)\big)$ holds  and 
$\overline{\tau}$ has a density with respect to $\mu$ bounded by $C=\max_i \frac{\alpha_i}{\mu(U_i)}\in [1,\infty)$. 
Therefore, $\mu$ can be represented as the convex combination $\mu=\frac{1}{C}\cdot \overline{\tau}+\left(1-\frac{1}{C}\right)\cdot \gamma,$ 
where $\gamma$ is a probability measure on $\Delta(\Omega)$. Let $\nu$ be the distribution that we get by condensing $\overline{\tau}$ to its center of masses $\mathbb{E}_{\overline{\tau}}[q]$ in this convex combination:
$$\nu=\frac{1}{C}\cdot\delta_{\mathbb{E}_{\overline{\tau}}[q]}+\left(1-\frac{1}{C}\right)\cdot \gamma.$$
Thus $\nu\preceq \mu$ and $\mathbb{E}_\mu\big[v_M(q) \big]< \mathbb{E}_{\nu}\big[v_M(q)\big]$. We conclude that $\mu$ is supported on a set that is not affine dominating cannot belong to ${\mathcal{M}}^0$.
\end{proof}

\begin{lemma}\label{lm_two_concavifications_are_the_same}
For  bounded upper semicontinuous $v_S$ and continuous $v_M$,
\begin{equation}\label{eq_concavification_as_maximization_over_FR}
   \cav_{{\mathcal{M}}^0}\big[v_S\big]=\cav_{{D}}\big[v_S\big].
\end{equation}
\end{lemma}

\begin{proof}
First, let us show that $\cav_{{\mathcal{M}}^0}\big[v_S\big]\geq \cav_{{D}}\big[v_S\big]$.
For continuous $v_M$, the set $D$ of collections $(q_1,\ldots,q_{|\Omega|})$ that are affine dominating with respect to $v_M$ is a non-empty compact subset of $\Delta(\Omega)^\Omega$. An upper semicontinuous function attains its maximum on a compact set and thus, for any prior $p\in \Delta(\Omega)$, there are $(q_1,\ldots,q_{|\Omega|})\in D$  and non-negative weights  $\alpha_1,\ldots, \alpha_{|\Omega|}$ summing up to one such that $\cav_{{D}}\big[v_S\big](p)=\sum_{k=1}^{|\Omega|}\alpha_k v_S(q_k)$ and $\sum_{k=1}^{|\Omega|}\alpha_k q_k=p$. Define a distribution $\mu\in \Delta(\Delta(\Omega))$ as a lottery over $(q_k)_{k=1,\ldots,|\Omega|}$ with corresponding weights $(\alpha_k)_{k=1,\ldots,|\Omega|}$, i.e., $\mu=\sum_{k=1}^{|\Omega|}\alpha_k \delta_{q_k}.$  By construction,  the support $\supp[\mu]$ coincides with $\{q_k,\ k=1,\ldots,|\Omega|\}$ and, hence, $\mu$ is supported on an affine dominating set. By Lemma~\ref{lm_FR_as_a_restriction_on_support}, the distribution $\mu$ belongs to the set ${\mathcal{M}}^0$. Since $\mu\in {\mathcal{M}}^0$ and $\mathbb{E}_\mu[q]=p$, we conclude that $\cav_{{\mathcal{M}}^0}\big[v_S\big]\geq \mathbb{E}_\mu[v_S(q)]$ by the definition of the constrained concavification~\eqref{eq_constrained_concavification_extended}. On the other hand,  $\mathbb{E}_\mu[v_S(q)]$ is equal to $\sum_{k=1}^{|\Omega|}\alpha_k v_S(q_k)$ and, hence, to $\cav_{{D}}\big[v_S\big](p)$. Thus $\cav_{{\mathcal{M}}^0}\big[v_S\big](p)\geq \cav_{{D}}\big[v_S\big](p)$.

Now let us prove that $\cav_{{\mathcal{M}}^0}\big[v_S\big](p)\leq \cav_{{D}}\big[v_S\big](p)$ for any $p\in \Delta(\Omega)$. It is enough to show that $\cav_{{\mathcal{M}}^0}\big[v_S\big](p)-\delta\leq \cav_{{D}}\big[v_S\big](p)$ for  any $\delta>0$. Fixing $\delta$, we find a distribution $\mu\in {\mathcal{M}}^0$ such that $\cav_{{\mathcal{M}}^0}\big[v_S\big]-\delta\leq \mathbb{E}_\mu[v_S(q)]$ and $\mathbb{E}_\mu[q]=p$. Consider the set $\mathcal{M}_p(\mu)$ of all distributions $\mu'\in \Delta(\Delta(\Omega))$  such that $\supp [\mu']\subset \supp[\mu]$ and $\mathbb{E}_{\mu'}[q]=p$. Then
$$ \max_{\mu'\in \mathcal{M}_p(\mu)} \mathbb{E}_{\mu'}[v_S(q)]\geq \cav_{{\mathcal{M}}^0}\big[v_S\big]-\delta.$$
The maximum is attained as $\mathcal{M}_p(\mu)$ is a compact set in the weak topology and $\mathbb{E}_{\mu'}[v_S(q)]$ is an upper semicontinuous functional of $\mu'$ for upper semicontinuous $v_S$. Moreover, $\mathcal{M}_p(\mu)$ is a convex set and so, by Bauer's principle, the maximum is attained at an extreme point of $\mathcal{M}_p(\mu)$. The extreme points of $\mathcal{M}_p(\mu)$ are\footnote{By Theorem~2.1 of \cite{winkler1988extreme}, the  extreme points of the set of all measures  satisfying $t$ linear constraints and defined on a general measurable space consist of convex combinations of at most $t+1$ point masses. The set $\mathcal{M}_p(\mu)$ can be seen as the set of all measures on $\supp[\mu]$ satisfying $t=|\Omega|-1$ scalar linear constraints: for all states $\omega$ except for one, $\mathbb{E}_{\mu'}[q(\omega)]=p(\omega)$ (the condition for the excluded state follows from other conditions as the total mass assigned by $p$ and $q$ is equal to one).} convex combinations of $|\Omega|$ point masses; i.e., they have the form $\sum_{k=1}^{|\Omega|}\alpha_k \delta_{q_k},$ where $q_k\in \supp[\mu]$ for all $k$ and $\sum_{k=1}^{|\Omega|}\alpha_k q_k=p$. By Lemma~\ref{lm_FR_as_a_restriction_on_support}, the measure $\mu$ is supported on an affine dominating set and, hence, the collection $(q_1,\ldots, q_n)$ is affine dominating;  i.e., it belongs to $D$. Thus 
$$\cav_D[v_S](p)\geq \sum_{k=1}^{|\Omega|}\alpha_k v_S(q_k)\geq \cav_{{\mathcal{M}}^0}\big[v_S\big](p)-\delta.$$
As this equality holds for any positive $\delta$ and $p$, we conclude that $\cav_D[v_S]\geq \cav_{{\mathcal{M}}^0}\big[v_S\big]$.

We checked that $\cav_D[v_S]\leq \cav_{{\mathcal{M}}^0}\big[v_S\big]$ and $\cav_D[v_S]\geq \cav_{{\mathcal{M}}^0}\big[v_S\big]$. Thus $\cav_D[v_S]$ is equal to $\cav_{{\mathcal{M}}^0}\big[v_S\big]$.
\end{proof}

As we show below, Theorem~\ref{theorem:appendix_one_mediatior} becomes a straightforward combination of the lemmas proved above.
\begin{proof}[Proof of Theorem~\ref{theorem:appendix_one_mediatior}]
For $\varepsilon>0$, Lemma~\ref{lm_upper_varepsilon} implies that the sender's optimal payoff $V_S^\varepsilon(p)$ in \ed{an $\varepsilon$-RTSPE} is at most $\cav_{{\mathcal{M}}^\varepsilon}\big[v_S\big](p)$ and Lemma~\ref{lm_lower_varepsilon} shows that it is at least this value. Thus $V_S^\varepsilon(p)=\cav_{{\mathcal{M}}^\varepsilon}\big[v_S\big](p)$ for any bounded measurable utilities and $\varepsilon>0$. We obtained the first statement of the theorem.

If $v_S$ is upper semicontinuous and $v_M$ is continuous, both lemmas allow us to set $\varepsilon=0$ and thus $V_S^0(p)=\cav_{{\mathcal{M}}^0}\big[v_S\big](p)$. Moreover, by Lemma~\ref{lm_lower_varepsilon}, the optimal equilibrium exists, i.e., the supremum over equilibria in the definition of $V_S^0(p)$ can be replaced by the maximum. By Lemma~\ref{lm_two_concavifications_are_the_same}, $\cav_{{\mathcal{M}}^0}\big[v_S\big]=\cav_{{D}}\big[v_S\big]$ and we get the second statement of the theorem which completes the proof.
\end{proof}

\section{Proof of Theorem~\ref{th_n_mediators}}\label{app_n_mediators}
We formulate and prove a generalization of Theorem~\ref{th_n_mediators} that does not impose any regularity assumption on indirect utilities. 
Recall that $v_S$, $v_{M_1},\ldots, v_{M_n}$ are indirect utilities of the sender and $n$ mediators. They are assumed to be bounded measurable functions on $\Delta(\Omega)$. Recall that, for $\varepsilon\geq 0$, the supremum of the sender's payoff  over all \ed{subgame perfect $\varepsilon$-equilibria with refined tie-breaking ($\varepsilon$-RTSPE) is denoted by $V_S^\varepsilon(p)$; see Appendix~\ref{app_perfect}.}

To characterize $V_S^\varepsilon(p)$,  define sets ${\mathcal{M}}_{i}^\varepsilon\subset \Delta(\Delta(\Omega))$ recursively. This definition extends the definition of ${\mathcal{M}}_{i}$ from~\eqref{eq_FRi_definition} that corresponds to $\varepsilon=0$ and the definition of ${\mathcal{M}}^\varepsilon$ from~\eqref{eq_FR_epsilon_one_mediator} that corresponds to  $n=1$. We define ${\mathcal{M}}_{n+1}^\varepsilon=\Delta(\Delta(\Omega))$ and
$${\mathcal{M}}_i^\varepsilon=\Big\{\mu\in {\mathcal{M}}_{i+1}^\varepsilon \ \Big|\  \Big(\nu\in {\mathcal{M}}_{i+1}^\varepsilon, \ \nu\preceq \mu\Big)\Longrightarrow   \mathbb{E}_\mu\big[v_{M_i}(q) \big]\geq \mathbb{E}_\nu\big[v_{M_i}(q) \big]-\varepsilon\Big\}.$$
To accommodate discontinuities, we replace maximum by supremum in the definition of constrained concavification~\eqref{eq_constrained_concavification_extended_body}: for $\mathcal{M}\subset \Delta(\Delta(\Omega))$, $$\cav_{\mathcal{M}}[v_S](p)=\sup\left\{ \mathbb{E}_\mu [v_S(q)] \  \Big| \ \mu\in\mathcal{M}, \  \mathbb{E}_\mu[q]=p \right\}.$$
\begin{theorem}[Generalized version of Theorem \ref{th_n_mediators}]\label{theorem:appendix_n_mediators}
Assuming that indirect utilities $v_S$, and $v_{M_1},\ldots, v_{M_n}$ are bounded measurable functions, we have that the sender's optimal payoff in \ed{a subgame perfect $\varepsilon$-equilibrium with refined tie-breaking} satisfies
$$V_S^\varepsilon(p)=\cav_{{\mathcal{M}}_1^\varepsilon}\big[v_S\big](p)$$ for any $\varepsilon>0$.

If, additionally, $v_S$ is upper semicontinuous and $v_{M_1},\ldots, v_{M_n}$ are continuous, the optimal equilibrium exists and $$V_S^0(p)=\cav_{{\mathcal{M}}^0}\big[v_S\big](p).$$
\end{theorem}
The proof generalizes the ideas presented in the proofs of Lemmas~\ref{lm_upper_varepsilon} and~\ref{lm_lower_varepsilon} to $n\geq 1$ mediators.
\begin{lemma}\label{lm_upper_varepsilon_n_mediators}
For $\varepsilon\geq 0$ and any \ed{$\varepsilon$-RTSPE}, the sender's payoff is upper-bounded by $\cav_{{\mathcal{M}}_1^\varepsilon}\big[v_S\big](p)$. 
\end{lemma}
\begin{proof}
Let $\mu$ be the distribution of the receiver's beliefs induced in this $\varepsilon$-equilibrium. The sender's payoff  equals $\mathbb{E}_{\mu}\big[v_S(q) \big]$. By the martingale property, $\mu$ has the right mean: $\mathbb{E}_{\mu}\big[q\big]=p$. Thus, to prove that $\mathbb{E}_{\mu}\big[v_S(q) \big]\leq \cav_{{\mathcal{M}}_1^\varepsilon}\big[v_S\big](p)$, it is enough to verify that $\mu$ belongs to ${\mathcal{M}}_1^\varepsilon$. By way of contradiction, assume that $\mu\notin {\mathcal{M}}_1^\varepsilon$. Since ${\mathcal{M}}_i^\varepsilon\subset {\mathcal{M}}_{i+1}^\varepsilon$ and ${\mathcal{M}}_{n+1}^\varepsilon=\Delta(\Delta(\Omega))$, we can find a mediator $M_i$ such that $\mu\notin {\mathcal{M}}_{i}^\varepsilon$ but $p\in {\mathcal{M}}_{i+1}^\varepsilon$. Let us construct a profitable deviation of mediator $M_i$.  By the definition of ${\mathcal{M}}_i^\varepsilon$, we see that there is $\nu\in {\mathcal{M}}_{i+1}^\varepsilon$ such that $\nu\preceq \mu$ and 
$\mathbb{E}_{\mu}\big[v_{M_i}(q) \big]<\mathbb{E}_{\nu}\big[v_{M_i}(q)] \big]-\varepsilon$. By Blackwell's theorem (Appendix~\ref{subsec_Blackwell}), the distribution $\mu_{M_i}$ of mediator $i$'s beliefs satisfies $\mu_{M_i} \succeq \mu$ and, hence, $\mu_{M_i}\succeq \nu$. Applying Blackwell's result again, we see that the mediator has a deviation that induces the distribution $\nu$ of $M_{i+1}$'s beliefs. Since $\nu\in {\mathcal{M}}_{i+1}^\varepsilon$, for any successor  $M_j$ of $M_{i+1}$ and any garbling $\nu'$ of $\nu$, we have 
$\mathbb{E}_{\nu}\big[v_{M_j}(q) \big]\geq\mathbb{E}_{\nu'}\big[v_{M_j}(q)] \big]-\varepsilon$; i.e., $M_j$ cannot benefit by more than $\varepsilon$ from any extra garbling. By the full-revelation refinement (Definition~\ref{def_subgame}), we conclude that all the mediators $M_j$ with $j\geq i+1$ will transmit $\nu$ as is.
 Thus $M_i$'s deviation improves her payoff from $\mathbb{E}_{\mu}\big[v_{M_i}(q) \big]$ to $\mathbb{E}_{\nu}\big[v_{M_i}(q) \big]$, i.e., by more than $\varepsilon$. This contradicts the assumption that we started from \ed{an $\varepsilon$-RTSPE}. This contradiction implies that $\mu$ necessarily belongs to ${\mathcal{M}}_1^\varepsilon$ and completes the proof.
\end{proof}
Note that the proof of Lemma~\ref{lm_upper_varepsilon_n_mediators} does not use the fact that the sender's signaling policy is an $\varepsilon$-best reply to the strategies of the other agents. Hence, a more general statement  holds: the sender's payoff is upper-bounded by $\cav_{{\mathcal{M}}_1^\varepsilon}\big[v_S\big](p)$ for any strategy of the sender and any collection of strategies of other agents forming \ed{an $\varepsilon$-RTSPE} in a subgame starting from the first mediator. This observation is needed to prove the next lemma.
\begin{lemma}\label{lm_lower_varepsilon_n_mediators}
For any $\varepsilon>0$ and any $\delta>0$, there exists \ed{an $\varepsilon$-RTSPE} such that the sender's payoff is at least $\cav_{{\mathcal{M}}_1^\varepsilon}\big[v_S\big](p)-\delta$.

If $v_S$ is  upper semicontinuous and bounded and $v_{M_1},\ldots, v_{M_n}$ are continuous, we can plug $\varepsilon=\delta=0$ into the above statement. 
\end{lemma}
\begin{proof}
We begin with proving the first statement: $\varepsilon>0$ and $\delta>0$ are fixed, the indirect utilities may be discontinuous, and our aim is to construct an $\varepsilon$-equilibrium with a sender's payoff of at least $\cav_{{\mathcal{M}}_1^\varepsilon}\big[v_S\big](p)-\delta$. Without loss of generality, we can assume that $\delta\leq \varepsilon.$

 By the definition of $\cav_{{\mathcal{M}}_1^\varepsilon}\big[v_S\big]$, we can find a distribution $\mu^\delta$  from ${\mathcal{M}}_1^\varepsilon$
with mean $p$ such that
$$\mathbb{E}_{\mu^\delta}[v_S(q)]\geq \cav_{{\mathcal{M}}_1^\varepsilon}\big[v_S\big](p)-\delta.$$
Consider the following profile of strategies. The sender selects a policy such that the belief of the first mediator is distributed according to $\mu^\delta$. Our intention is to define the strategies of the mediators such that the sender's signal is not garbled on the equilibrium path. Each mediator $M_i$ 
computes the distribution $\mu_{i}$ of her belief $p_{i}$ induced by the strategies of her predecessors. If $\mu_{i}$ belongs to ${\mathcal{M}}_i^\varepsilon$, the mediator transmits the signal unchanged to the next agent in the line. If $\mu_{i}\notin {\mathcal{M}}_i^\varepsilon$, the mediator can find a garbling improving her payoff by more than $\varepsilon$. Namely, there exists $\nu_i^\varepsilon(\mu_i)\in {\mathcal{M}}_{i}^\varepsilon$ such that $\nu_i^\varepsilon(\mu_i)\preceq \mu_i$ and the utility of the mediator $\mathbb{E}_{\nu_i^\varepsilon(\mu_i)}[v_{M_i}(q)]$ is within $\varepsilon$  from $\sup\{\mathbb{E}_{\nu}[v_{M_i}(q)],\ \nu\in {\mathcal{M}}_i^\varepsilon,\ \nu \preceq \mu_i\}$. For $\mu_i\notin {\mathcal{M}}_i^\varepsilon$, mediator $M_i$ selects a policy inducing the distribution of beliefs $\nu_i^\varepsilon(\mu_i)$ of the next agent.

Since $\mu^\delta$ belongs to all ${\mathcal{M}}_i^\varepsilon$, no mediator garbles the signal of the sender. Thus the sender's payoff is  $\mathbb{E}_{\mu^\delta}[v_S(q)]r \geq \cav_{{\mathcal{M}}_1^\varepsilon}\big[v_S\big](p)-\delta$. It remains to check that the constructed profile of strategies is indeed  \ed{ an  $\varepsilon$-RTSPE.} By the definition of the mediators' strategies, it is immediate that the profile satisfies the full-revelation refinement. Hence, only the absence of $\varepsilon$-profitable deviations needs to be checked.

Let us show that no agent has a deviation improving her payoff by more than $\varepsilon$. First, consider a deviation of a mediator $M_i$ with $i<n$. There are two cases depending on whether the distribution of the next mediator's beliefs $\mu_{i+1}$ induced by the deviation belongs to ${\mathcal{M}}_{i+1}^\varepsilon$ or not. If $\mu_{i+1}\in {\mathcal{M}}_{i+1}^\varepsilon$, then
all the subsequent mediators will transmit the signal unchanged and so $\mu_{i+1}$ propagates to the receiver. Hence, the deviator's payoff equals  $\mathbb{E}_{\mu_{i+1}}\big[v_{M_i}(q) \big]$. By the definition of ${\mathcal{M}}_i^\varepsilon$ and the fact that $\mu_{i+1}\preceq \mu^\delta\in {\mathcal{M}}_{i}^\varepsilon$, we conclude that
$\mathbb{E}_{\mu_{i+1}}\big[v_{M_i}(q) \big] \leq \mathbb{E}_{\mu^\delta}\big[v_{M_i}(q) \big]+\varepsilon$; i.e., the deviation cannot increase $M_i$'s  payoff by more than $\varepsilon$. The second possibility is that $\mu_{i+1}\notin {\mathcal{M}}_{i+1}^\varepsilon$. The strategy of $M_{i+1}$  prescribes that she garble, thereby inducing the distribution $\mu_{i+2}=\nu_{i+1}^\varepsilon(\mu_{i+1})$ of the next agent's beliefs. Such $\mu_{i+2}$ is contained in ${\mathcal{M}}_{i+1}^\varepsilon$ and so is not garbled by the successors. Thus the payoff of the deviator is $\mathbb{E}_{\mu_{i+2}}\big[v_{M_i}(q) \big]$, where $\mu_{i+2}\in {\mathcal{M}}_{i+1}^\varepsilon$ and $\mu_{i+2}\preceq \mu^\delta$. Once again, the definition of ${\mathcal{M}}_i^\varepsilon$ implies that $M_i$ cannot benefit by more than $\varepsilon$ from the deviation. The argument showing that the last mediator $M_n$ cannot improve her payoff by more than $\varepsilon$ is similar but simpler and, therefore, omitted.  

We conclude that, in a subgame starting from the first mediator, the mediators' strategies form \ed{an $\varepsilon$-RTSPE.} By Lemma~\ref{lm_upper_varepsilon_n_mediators} and the discussion after its proof, the sender's payoff is upper-bounded by $\cav_{{\mathcal{M}}_1^\varepsilon}\big[v_S\big](p)$ for any her deviation. Since the sender's original strategy gives her a payoff of at least $\cav_{{\mathcal{M}}_1^\varepsilon}\big[v_S\big](p)-\delta$ and $\delta\leq \varepsilon$, the sender has no deviations improving  her payoff by more than $\varepsilon$. Thus the constructed profile of strategies is \ed{an $\varepsilon$-RTSPE} with a sender's payoff of at least $\cav_{{\mathcal{M}}_1^\varepsilon}\big[v_S\big](p)-\delta$, which proves the first statement of the lemma.

Let us prove the second statement of the lemma, which allows us to set $\varepsilon=\delta=0$ for bounded upper semicontinuous $v_S$ and continuous $v_{M_1},\ldots, v_{M_n}$. By the continuity of the mediators' utilities, the sets ${\mathcal{M}}_i^\varepsilon$ are compact subsets of $\Delta(\Delta(\Omega))$ endowed with the weak topology for any $\varepsilon\geq 0$. By the upper semicontinuity of $v_S$, the integral $\mathbb{E}_\mu[v_S(q)]$ is an upper semicontinuous functional of $\mu\in \Delta(\Delta(\Omega))$. Hence, there exists the sender's optimal distribution  $\mu^\delta$ with $\delta=0$ since an upper semicontinuous functional attains its maximum on a compact set. For the same reason, the mediators' replies $\nu_i^\varepsilon$ exist for $\varepsilon=0$. We conclude that we can plug $\varepsilon=\delta=0$, resulting in \ed{an RTSPE} with a sender's payoff of $\cav_{{\mathcal{M}}_1^0}\big[v_S\big](p)$.
\end{proof}
Theorem~\ref{theorem:appendix_n_mediators}  follows from these two lemmas.
\begin{proof}[Proof of Theorem~\ref{theorem:appendix_n_mediators}.]
From Lemma~\ref{lm_upper_varepsilon_n_mediators}, $V_S^\varepsilon(p)\leq \cav_{{\mathcal{M}}_1^\varepsilon}\big[v_S\big](p)$ for any $\varepsilon\geq 0$. The first part of Lemma~\ref{lm_lower_varepsilon_n_mediators} gives the opposite inequality for $\varepsilon>0$. We obtain $V_S^\varepsilon(p)= \cav_{{\mathcal{M}}_1^\varepsilon}\big[v_S\big](p)$ for $\varepsilon>0$, which completes the proof of the first statement of the theorem.

If the utilities are regular, a \ed{an RTSPE} with a sender's payoff of $\cav_{{\mathcal{M}}_1^0}\big[v_S\big](p)$ exists by the second part of  Lemma~\ref{lm_lower_varepsilon_n_mediators}. Combining this with the upper bound from Lemma~\ref{lm_upper_varepsilon_n_mediators}, we obtain $V_S^0(p)= \cav_{{\mathcal{M}}_1^0}\big[v_S\big](p)$. This completes the proof of the  second statement of the theorem.
\end{proof}

\ifdefined\AEJ
\newpage

\begin{center}
    {\Huge Online Appendix }
\end{center}
\fi

\section{Sequential Games over Partially Ordered Sets}\label{sect_partial_order}
An analog of the recursive representation of the sender's optimal payoff obtained in Theorem~\ref{th_n_mediators} can be proved for a broad class of games, where agents move a token sequentially over a partially ordered set and the payoffs are determined by the final position of the token. This provides a unifying perspective on our results and the results for multiple-sender models of Li and Norman~\cite{li2018sequential} and Wu~\cite{wu2020essays}.

Let $X$ be a compact set endowed with a continuous partial order $\succeq$. 
A token is originally placed at a point $x_0\in X$. Agents $i=0,1,\ldots, n$ move the token sequentially. Agent~$i$ can move it from $x_i$ to any point $x_{i+1}=F_i(x_i)$ such that $x_i\succeq x_{i+1}$. The map $F_i$ is agent $i$'s policy. The payoffs are determined by the final position of the token $x_{n+1}$. The payoff to agent $i$ is given by $w_i(x_{n+1}),$ where $w_i$ is a continuous utility function $w_i\colon X\to\mathbb{R}$.

Agents select their policies sequentially and so the choice of agent $i$'s policy $F_i$ can be affected by the history of choices $h_i=(F_0,\ldots,F_{i-1})$. A subgame perfect equilibrium is defined in the standard way \ed{and the following property is an analog of the full-revelation refinement: if the identity map $\mathrm{id}(x)=x$ is a best reply to a history $h_i$, then agent $i$ selects a policy $F_i=\mathrm{id}$ at this history. We will refer to such equilibria as subgame perfect equilibria with refined tie-breaking.}\footnote{\ed{An analog of the belief-driven receiver refinement is hardwired in the assumption that agents' actions are determined by the current position of the token.}}

We define sets $X_i$ recursively so  that an agent $i$ has no incentive to move the token whenever $x_i\in X_i$. Let $X_{n+1}=X$ and
$$X_i=\Big\{x\in X_{i+1}\, \big| \, \big(x'\in X_{i+1},\, x'\preceq x\big)\Longrightarrow w_i(x)\geq w_i(x')\Big\}.$$
\begin{theorem}
The maximal payoff that agent $0$ can achieve in a \ed{subgame perfect equilibrium with refined tie-breaking} is equal to
$$\max\Big\{w_0(x)\mid x\in X_1, \ x\preceq x_0\Big\}.$$
\end{theorem}
The proof mimics that of Theorem~\ref{th_n_mediators} and is, therefore, omitted. Similarly, to that theorem, the optimum is achieved in an equilibrium where agent~$0$ moves the token to $x_1=\argmax\big\{w_0(x)\mid x\in X_1, \ x\preceq x_0\big\}$ and the other agents do not move it anymore, i.e., a version of the revelation principle holds.
\medskip

Our persuasion model with mediators can be reduced to a version of this game where $X=\Delta(\Delta(\Omega))$, the comparison $\mu\succeq\nu$ means that $\nu$ is a mean-preserving contraction of $\mu$, utilities are given by 
$w_i(\mu)=\int_{\Delta(\Omega)}v_i(q)\dd\mu(q)$, and the initial point $x_0$ corresponds to the belief $\mu_0$ about $\omega$ induced by the realization of $\omega$, i.e., $\mu_0=\sum_{k\in \Omega} p_k\cdot \delta_{\delta_k}$.

Similarly, the models of  Li and Norman~\cite{li2018sequential} and Wu~\cite{wu2020essays}, where the senders move sequentially adding more and more information, correspond to reversing the partial order defined above; i.e.,  $\mu\succeq\nu$ if $\mu$ is a mean-preserving contraction of $\nu$. The starting point $x_0$ represents having no information about $\omega$, i.e., $\mu_0=\delta_p$.

\section{Persuasion on Networks}\label{sect_two_direct_mediators}

Our model of mediated persuasion with a sequence of mediators can be seen as an example of persuasion over networks, where the network is just the line graph. A natural next step would be to understand persuasion over rooted tree graphs, where the sender is at the root, the receivers taking actions are located  at the leaves, and the rest of the nodes are mediators transmitting the information received from predecessors to successors, while possibly garbling it. 

The main obstacle arising for general networks
is the failure of the revelation principle that underpins our analysis. We illustrate this obstacle in an example and leave the analysis of persuasion with general networks for future research.

Consider a persuasion problem with the simple tree network depicted in Figure \ref{fig:net}, where the sender $S$ communicates directly and publicly with two mediators $M_1$ and $M_2$, who in turn communicate the information to receivers $R_1$ and $R_2$, respectively.
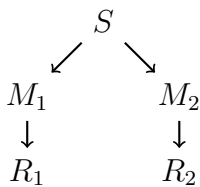
\begin{figure}[H]
    \centering
    \caption{The network}
    \label{fig:net}
    \begin{tikzpicture}
        \node(S) at (0,2) {$S$};
        \node(M1) at (-1,1) {$M_1$};
        \node(M2) at (1,1) {$M_2$};
        \node(R1) at (-1,0) {$R_1$};
        \node(R2) at (1,0) {$R_2$};
        \draw[thick, ->] (S)--(M1);
        \draw[thick,->] (S)--(M2);
        \draw[thick,->] (M1)--(R1);
        \draw[thick,->] (M2)--(R2);
    \end{tikzpicture}
\end{figure}

We will see that this example exhibits two phenomena. First, the sender's optimal payoff cannot be reached in the class of equilibria, where none of the mediators garbles the received signal; i.e., the revelation principle fails. Second, the number of signals used by the sender in the optimal equilibrium exceeds $|\Omega|$. In Example~\ref{eq_FR1_example}, we already observed that optimal persuasion with two or more mediators in a line may require more than $|\Omega|$ signals; now we see that this is the case even if there is just one mediator between the sender and each receiver and the two receivers are persuaded simultaneously.


For simplicity, we assume that the prior $p=\frac{1}{2}$, the utility of mediator $M_i$ depends only on the action of receiver $R_i$, and the sender's utility is an additively separable function of
the actions of $R_1$ and $R_2$. 
 The corresponding indirect utilities are  $v_{M_1}(p_1)$, $v_{M_2}(p_2)$ and $v^1_S(p_1)+v^2_S(p_2)$, where 
 $p_1$ and $p_2$ are the beliefs  of receivers $R_1$ and $R_2$, respectively. The indirect 
utilities are given in Figure \ref{fig:ut}. 
\begin{figure}[h]
    \centering
    \caption{The indirect utilities.
    \label{fig:ut}}

\begin{tikzpicture}[scale=0.6]
\pgfplotsset{every tick/.append style={semithick, major tick length=5pt, minor tick length=5pt, black}};
\pgfplotsset{every tick label/.append style={font=\Large}};
    \begin{axis}[
    axis on top=true,
    axis line style = ultra thick,
    axis x line=center,
    xtick={0, 0.2, 0.55, 0.75, 1},
    xlabel=${p_1}$,
    x label style={at={(1.03,0.02)}, font=\huge},
    axis y line=center,
    ylabel=$v^1_S$,
    y label style={at={(-0.35,1.1)}, font=\huge},
    no marks]
    \addplot+[smooth,ultra thick,blue,name path=A,domain=0:0.2] {1-5*x};
    \addplot+[smooth,ultra thick,blue,name path=A,domain=0.2:0.55] {0.005};
    \addplot+[smooth,ultra thick,blue,name path=A,domain=0.55:0.75] {5*(x-0.55)};
    \addplot+[smooth,ultra thick,blue,name path=A,domain=0.75:1] {1-0.5*(x-0.75)};
    \end{axis}
\end{tikzpicture}
\hskip 1cm
\begin{tikzpicture}[scale=0.6]
\pgfplotsset{every tick/.append style={semithick, major tick length=5pt, minor tick length=5pt, black}};
\pgfplotsset{every tick label/.append style={font=\Large}};
    \begin{axis}[
     axis on top=true,
    axis line style = ultra thick,
    axis x line=center,
    xtick={0, 0.25, 0.45, 0.8, 1},
    xlabel=$p_2$,
    x label style={at={(1.03,0.02)}, font=\huge},
    axis y line=center,
    ylabel=$v^2_S$,
    y label style={at={(-0.35,1.1)}, font=\huge},
    no marks]
    \addplot+[smooth,ultra thick,blue,name path=A,domain=0:0.25] {1+0.5*(x-0.25)};
    \addplot+[smooth,ultra thick,blue,name path=A,domain=0.25:0.45] {1-5*(x-0.25)};
    \addplot+[smooth,ultra thick,blue,name path=A,domain=0.45:0.8] {0.005};
    \addplot+[smooth,ultra thick,blue,name path=A,domain=0.8:1] {5*(x-0.8)};
    \end{axis}
\end{tikzpicture}
\vskip 0.5cm
\begin{tikzpicture}[scale=0.6]
\pgfplotsset{every tick/.append style={semithick, major tick length=5pt, minor tick length=5pt, black}};
\pgfplotsset{every tick label/.append style={font=\Large}};
    \begin{axis}[
     axis on top=true,
    axis line style = ultra thick,
    axis x line=center,
    xtick={0, 0.2, 0.55, 0.75, 1},
    xlabel=$p_1$,
    x label style={at={(1.03,0.02)}, font=\huge},
    axis y line=center,
    ylabel=$v_{M_1}$,
    y label style={at={(-0.35,1.1)}, font=\huge},
    no marks]
    \addplot+[smooth,ultra thick,blue,name path=A,domain=0:0.2] {1-5*x};
    \addplot+[smooth,ultra thick,blue,name path=A,domain=0.2:0.55] {0.005};
    \addplot+[smooth,ultra thick,blue,name path=A,domain=0.55:0.75] {5*(x-0.55)};
    \addplot+[smooth,ultra thick,blue,name path=A,domain=0.75:1] {1+0.5*(x-0.75)};
    \end{axis}
\end{tikzpicture}
\hskip 1cm
\begin{tikzpicture}[scale=0.6]
\pgfplotsset{every tick/.append style={semithick, major tick length=5pt, minor tick length=5pt, black}};
\pgfplotsset{every tick label/.append style={font=\Large}};
    \begin{axis}[
     axis on top=true,
    axis line style = ultra thick,
    axis x line=center,
    xtick={0, 0.25, 0.45, 0.8, 1},
    xlabel=$p_2$,
    x label style={at={(1.03,0.02)}, font=\huge},
    axis y line=center,
    ylabel=$v_{M_2}$,
    y label style={at={(-0.35,1.1)}, font=\huge},
    no marks]
    \addplot+[smooth,ultra thick,blue,name path=A,domain=0:0.25] {1-0.5*(x-0.25)};
    \addplot+[smooth,ultra thick,blue,name path=A,domain=0.25:0.45] {1-5*(x-0.25)};
    \addplot+[smooth,ultra thick,blue,name path=A,domain=0.45:0.8] {0.005};
    \addplot+[smooth,ultra thick,blue,name path=A,domain=0.8:1] {5*(x-0.8)};
    \end{axis}
\end{tikzpicture}

\end{figure}
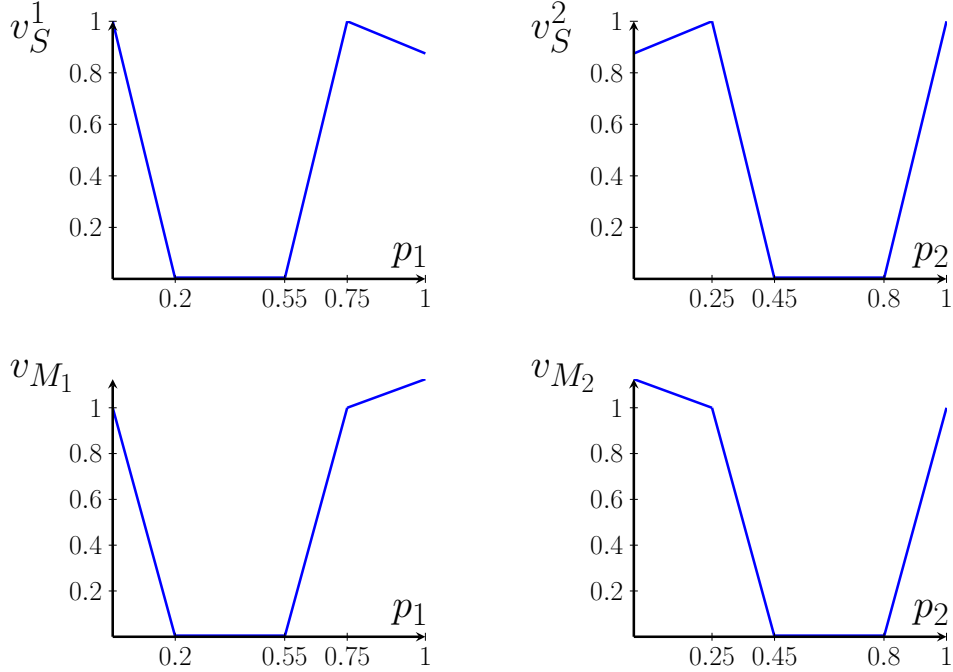

The sender's communication with {mediators $M_1$ and $M_2$} is public---i.e., both mediators observe the same signal sent by the sender---which does not allow us to split the problem into two independent persuasion problems.\footnote{If {we considered} private instead of public communication, then, due to the additive separability of sender's utility, the problem would reduce to a pair of one-mediator persuasion problems: the problem with mediator $i$ and receiver $i$ for $i=1,2$.
We enforce the interconnectedness of the two problems by assuming public signals since the alternative assumption of non-separable utility would make the problem intractable. Indeed, in the private-communication setting with non-separable utility, persuasion problems are extremely involved even without mediators; see \cite{arieli2021feasible}.} 

Let us show first {that} the sender can achieve  her ideal utility of $2$; i.e., {she can induce the belief $\p_1\in \big\{0,\,\frac{3}{4}\big\}$ of receiver $R_1$ and the belief $p_2\in \big\{\frac{1}{4},\,1\big\}$ of receiver $R_2$} with probability $1$. To do this, the sender uses a {ternary-signal} policy {that} induces one of the three belief $\big\{0,\frac{1}{2},1\big\}$ {of the mediators} with equal probabilities of $\frac{1}{3}$. {The best reply of mediator $M_1$ to this policy}  is to garble the signal by pooling together the posteriors $\frac{1}{2}$ and $1$ into the posterior of $\frac{3}{4}$ and {keep} the posterior $0$ {unchanged}. {This policy of mediator $M_1$ induces the posterior  $p_1=\frac{3}{4}$ of receiver $R_1$ with probability $\frac{2}{3}$ and the posterior~$p_1=0$, with probability~$\frac{1}{3}$. Similarly, mediator~$M_2$ pools together the posteriors $0$ and $\frac{1}{2}$ and reveals the posterior $1$ as it is; by this policy, she induces beliefs $p_2=\frac{1}{4}$ and $p_2=1$ of receiver~$2$ with probabilities $\frac{2}{3}$ and $\frac{1}{3}$.}

{One can check that mediators' policies are  best replies using the technique of price functions introduced by {Dworczak and Martini}~\cite{dworczak2019simple}. Let us sketch the argument for mediator~$M_2$; the argument for the second mediator is symmetric and, therefore, omitted. Let $\mu_1$ be the distribution of the beliefs of~$M_1$, induced by the sender's policy; i.e, $\mu_{M_1}$ is the uniform distribution over $\big\{0,\frac{1}{2},1\big\}$. For any policy of $M_1$, the induced distribution of beliefs $\mu_{R_1}$ of receiver $R_1$ is a mean-preserving contraction of $\mu_{M_1}$, i.e., $\mu_{M_1}\succeq \mu_{R_1}$; see Appendix~\ref{subsec_Blackwell}. 
 Therefore, for any convex function $\pi_{M_1}\geq v_{M_1}$ and any policy of $M_1$, her payoff cannot exceed the expectation of $\pi_{M_1}$ with respect to $\mu_{M_1}$; indeed, a mean-preserving contraction can only decrease the expected value of a convex function. Consequently, if for a given policy of the mediator we can find a function $\pi_{M_1}$ such that this upper bound coincides with the payoff guaranteed by the policy (the expected value of $v_{M_1}$), then this policy is a best reply. One can verify that the piecewise linear function 
\begin{align*}
    \pi_{M_1}(q)=\begin{cases}
    1-q, &\text{ for } q\in [0,\frac{1}{2}] \\
    2q-\frac{1}{2}, &\text{ for } q\in [\frac{1}{2},1]
    \end{cases}
\end{align*}
satisfies these requirements for the policy inducing the pair of posteriors $p_1\in\big\{0,\frac{3}{4}\big\}$ with probabilities $\frac{1}{3}$ and $\frac{2}{3}$ respectively; thus this policy is a best reply for mediator~$M_1$.}


{We stress that the described equilibrium exhibits a peculiar} phenomenon: the sender  
provides the mediators with {partial information specially tailored to their incentives} and each mediator garbles this partial information in a way that is ideal for the sender. Interestingly, this {partial} information uses {three} different signals, in contrast to persuasion with one mediator, or no mediators where binary signals are sufficient.

We can now demonstrate that binary-signal policies {are not enough for the} sender {to extract her ideal} utility of~$2$. {Consider a binary-signal policy and denote by $q\leq q'$ the pair of the mediators' posteriors induced by this policy.} We consider several cases. If $q>0$, then receiver $R_1$ {cannot get} a posterior of $0$, and hence the sender {does not} obtain her optimal payoff because of  $R_1$. Similarly, if $q'<1$,  receiver $R_2$ {cannot have} a posterior of $1$ {and so the sender again gets a suboptimal payoff.} The only remaining policy is the full-revelation one (i.e., $q=0$ and $q'=1$). {For this} policy, both mediators fully reveal the information to the receivers; this {again} results in {a suboptimal payoff to} the sender. 
\end{document}